%% file: arxiv-paper.tex
\PassOptionsToPackage{table}{xcolor}

\documentclass[acmsmall,screen,nonacm]{acmart}

\AtBeginDocument{%
  \providecommand\BibTeX{{%
    \normalfont B\kern-0.5em{\scshape i\kern-0.25em b}\kern-0.8em\TeX}}}

\setcopyright{rightsretained}
\copyrightyear{2026}

\input{macros}

\arxivtrue

\begin{document}

\title{\flowbook: Dynamic Enforcement of Output Consistency in Computational Notebooks}

\author{Stephen N. Freund}
\email{freund@cs.williams.edu}
\orcid{0009-0000-6992-199X}
\affiliation{%
  \institution{Williams College}
  \country{Williamstown, MA, USA}
}

\author{Emery D. Berger}
\authornote{Work done at the University of Massachusetts Amherst.}
\email{emery@cs.umass.edu}
\orcid{0000-0002-3222-3271}
\affiliation{%
  \institution{University of Massachusetts Amherst}
  \country{USA}
}
\affiliation{%
  \institution{Amazon Web Services}
  \country{USA}
}

\author{Cormac Flanagan}
\email{cormac@ucsc.edu}
\orcid{0009-0009-5067-6774}
\affiliation{%
  \institution{University of California, Santa Cruz}
  \country{Santa Cruz, CA, USA}
}

\author{Eunice Jun}
\email{emjun@cs.ucla.edu}
\orcid{0000-0002-4050-4284}
\affiliation{%
  \institution{University of California, Los Angeles}
  \country{Los Angeles, CA, USA}
}

\renewcommand{\shortauthors}{Freund, Berger, Flanagan, and Jun}

\input{abstract}

\maketitle

\input{introduction}
\input{example}
\input{flowbook-analysis}
\input{implementation}
\input{evaluation}

\input{related}

\input{threats}
\input{conclusion}
\input{acknowledgements}

\bibliographystyle{ACM-Reference-Format}
\bibliography{main,steve,emery,emjun}

\clearpage
\appendix
\section*{Appendix}
\input{supplement-body}

\end{document}

%% file: macros.tex
\newif\ifarxiv
\arxivfalse

\usepackage[normalem]{ulem} %
\usepackage{fontawesome5}
\usepackage{xspace}
\usepackage{xcolor}
\usepackage{listings}
\usepackage{graphicx}
\usepackage{microtype}
\usepackage{multirow}
\usepackage{array}
\usepackage{enumitem}
\usepackage{mathtools}
\usepackage{tabularx}
\usepackage{tcolorbox}
\tcbuselibrary{skins}
\input{jupyter-macros.tex}
\usepackage{algorithm}
\usepackage{algpseudocode}
\usepackage{wrapfig}

\usepackage{amsmath}
\usepackage{booktabs}
\usepackage{mathpartir}
\usepackage{subcaption}

\newcommand{\steve}[1]{\textcolor{red}{#1}}
\renewcommand{\steve}[1]{#1}

\newcommand{\note}[1]{#1}
\definecolor{datared}{RGB}{255,0,0}
\newcommand{\data}[1]{\textcolor{datared}{#1}}
\renewcommand{\data}[1]{#1}

\newcommand{\punt}[1]{}

\newcommand{\relsize}{\small}

\newcommand{\rulegap}{\\~\\}
\newcommand{\relname}[1]{\raisebox{0ex}[0ex][0ex]{\mbox{\relsize \textsc{[#1]}}}}
\newcommand{\myheadrule}[1]{\fbox{\ensuremath{#1}}\vspace{1mm}}

\newcommand{\defeq}{\stackrel{\text{def}}{=}}

\newcommand{\nrln}[3]{
\mbox{
$\begin{array}{@{}l@{}}
\relname{#1} \\
~~\frac
    {\strut\begin{array}{c} #2 \end{array}}
    {\strut\begin{array}{c} #3 \end{array}}
\end{array}
$}}

\newcommand{\Loc}{\mathrm{Loc}}
\newcommand{\Val}{\mathrm{Val}}

\newcommand{\Edit}{\textsc{Edit}}
\newcommand{\Run}{\textsc{Run}}
\newcommand{\Insert}{\textsc{Insert}}
\newcommand{\Delete}{\textsc{Delete}}
\newcommand{\Move}{\textsc{Move}}

\newcommand{\cat}{,\,}

\newcommand{\StdEvalCell}[4]{#1;\, #2 \,\downarrow\, #3 \cdot #4}

\newcommand{\StdEvalNotebook}[3]{#1 \,\downarrow\, #2 \cdot #3}

\newcommand{\InstEvalCell}[6]{#1;\, #2 \,\Downarrow\, #3 \cdot #4 \cdot #5 \cdot #6}

\newcommand{\clean}{\textsc{clean}}
\newcommand{\stale}{\textsc{stale}}

\newcommand{\NoReadAndWrite}{\mbox{\textsc{NoReadAndWrite}}}
\newcommand{\WriteBeforeRead}{\mbox{\textsc{WriteAboveRead}}}
\newcommand{\NoReadBeforeWrite}{\mbox{\textsc{NoReadAboveWrite}}}
\newcommand{\NoWriteAfterRead}{\mbox{\textsc{NoWriteBelowRead}}}

\newcommand{\LastWriter}{\mbox{\textsc{LastWriter}}}

\newcommand{\ForwardStale}{\textsc{ForwardStale}}
\newcommand{\BackwardStale}{\textsc{BackwardStale}}

\newcommand{\flowbook}{\textsc{FlowBook}\xspace}

\newcommand{\myparagraph}[1]{\par\smallskip\noindent\textbf{#1.}}

\def\<#1>{\codeid{#1}}
\newcommand{\codeid}[1]{\ifmmode{\mbox{\small\ttfamily{#1}}}\else{\small\ttfamily #1}\fi}

%% file: abstract.tex
\begin{abstract}

Computational notebooks interleave code cells with their outputs and
let users edit and execute cells in arbitrary order.  This flexibility
supports interactive exploration and iterative development, but it
also allows outputs to become inconsistent with the notebook's current
code and data. Such inconsistencies can lead users to draw invalid scientific conclusions from
inconsistent outputs. We argue that notebooks should always be
\emph{output consistent}: every output should match what a
top-to-bottom execution of the current code would produce. Forcing a
top-to-bottom execution after each edit would trivially guarantee
this, but at prohibitive cost: such an execution can take minutes to
hours for expensive computations.

We present \flowbook{}, a dynamic analysis for Python notebooks that
enforces output consistency without giving up out-of-order execution.
\flowbook{} does this by enforcing a stronger property, \emph{rerun
  consistency}: no location read by a cell may be written by any cell
below it.  Rerun consistency guarantees that rerunning a cell later,
in any order, reproduces its current output.  \flowbook{} tracks each
cell's read and write sets at cell boundaries, blocks updates that
would violate rerun consistency, and marks as \emph{stale} any cell
whose inputs may have changed. We prove that a notebook with no stale
cells is both rerun and output consistent.  Users remain free to
execute cells in arbitrary order; the only restriction is that
recorded outputs may not silently depend on the particular order in
which cells were executed.

Across 61 real-world notebooks, \flowbook{} flagged 203
rerun-inconsistent cells.  Rerunning 62\% led to output inconsistencies.
A further 21\% change notebook state with no visible output, which
can affect the results of later edits and reruns; only 5\% are false
positives of \flowbook{}'s conservative analysis.  \flowbook{}'s AI
repair agent repaired all 203 cells with local edits and no human
intervention.  \flowbook{} adds a median overhead of only 26\,ms per
cell execution.

\end{abstract}

%% file: introduction.tex
\section{Introduction}
\label{sec:introduction}

\input{nbs/calibration-litmus-tests.tex}

A computational notebook interleaves code cells with the outputs they
produced, and its interactive execution model lets users edit and run
cells in any order.  While this flexibility makes notebooks effective
for exploration, the combination of mutable shared state and
non-linear cell execution also allows recorded outputs to diverge from
what a top-to-bottom execution of the current code would
produce~\cite{DBLP:conf/chi/ChattopadhyayPH20,DBLP:conf/chi/KeryRAJM18,pimentel2021understanding,DBLP:conf/chi/RuleTH18,DBLP:journals/pacmhci/RuleDTH18}.
We say a notebook is \emph{output consistent} when its recorded
outputs match those of a top-to-bottom execution of the current cells
from an empty store; otherwise, it is \emph{output
  inconsistent}. Standard Jupyter notebooks do not alert users when
notebooks become output inconsistent.  As a result, users may draw incorrect scientific conclusions based
on outputs that no longer reflect the current code and data.

Figure~\ref{fig:calibration-litmus-tests} makes this hazard concrete
with a small air-quality analysis.  In the first panel, a scientist
loads a city's raw sensor readings of PM2.5, a measure of airborne
particulate matter less than 2.5 microns in diameter, in cell
\cell{A}; computes the mean raw reading (cell \cell{B}); and plots the
readings against the EPA health bands of \emph{Good}, \emph{Moderate},
or \emph{Unhealthy (for sensitive groups)} (cell \cell{D}).  The
scientist discovers that the sensor's raw readings were not properly
calibrated and inserts a new cell \cell{C} to apply a calibration
gain. After running cell \cell{C}, the plot in cell \cell{D} is output
inconsistent, but no visual indication is given in standard Jupyter
notebooks.  Rerunning \cell{D} would repair it, but the scientist may
not notice the inconsistency.  The second panel shows the notebook
after the scientist performs a top-to-bottom execution.  In this case,
the plot is consistent with the calibrated readings.  The third panel
extends that execution with the scientist revising the calibration
gain from 1.8 to 1.85 in \cell{C} and rerunning it. Since the
DataFrame is updated in place, the new gain compounds with the old
adjustment (1.8) rather than replacing it.  Rerunning \cell{D} plots
readings well into the \emph{Unhealthy} range.  That plot is output
inconsistent, and rerunning \cell{D} does not repair it: every rerun
of \cell{D} plots the compounded values. %

\data{A top-to-bottom execution after every edit would ensure output
  consistency, but this brute-force strategy degrades latency and
  compromises interactivity, as such executions often take minutes to
  hours in computationally expensive work (e.g., loading a large
  dataset, training a model).}
  
Prior approaches seek to enforce output consistency via reactive
dataflow models that re-execute cells automatically from inferred
dependencies~\cite{marimo2026,plutojl2026}.  Unfortunately, they
either integrate poorly with Python's scientific ecosystem or fail to
track dependencies precisely, leading to missed dependencies, output
inconsistencies, or unnecessary re-execution~\cite{arXiv:2511.21994}.
For example, Marimo's dependency analysis does not track or reason
about in-place mutations of DataFrames, meaning it fails to identify
the dependency between \cell{C} and \cell{D} and may rerun \cell{C} in
the third panel without mitigating the compound gain.

Other approaches try to reason more precisely about dependencies via
static and dynamic dependency
analyses~\cite{DBLP:conf/chi/HeadHBDD19,DBLP:journals/pvldb/ShankarMCHP22}.
These analyses are limited by incomplete dependency tracking and by
performance overhead.  Further, these systems do not enforce output
consistency.  For example,
IPyflow~\cite{DBLP:journals/pvldb/ShankarMCHP22} tracks most mutations
to DataFrame columns, but it does not prevent rerunning \cell{C} in
the third panel.

This paper exploits the insight that it is possible to enforce output
consistency without precise dependency inference for all memory
locations. Instead, we enforce \emph{rerun consistency}, where the
locations read by one cell are not modified by any cell later in the
notebook.  We present \flowbook{}, a Jupyter plugin for Python
notebooks that enforces rerun consistency.  \flowbook{} employs a
lightweight dynamic analysis: it tracks read and write sets at cell
boundaries, prevents code violating rerun consistency from leaving the
notebook in a state that no top-to-bottom execution could produce, and
flags as \emph{stale} any cell that reads a location whose value may
have changed. A notebook with no stale cells is guaranteed to be
output consistent (Section~\ref{sec:analysis}).

Rerun and output consistency are inspired by consistency models for
databases and distributed systems, such as serializability and
consistent prefix semantics, but differ from them in non-trivial ways,
such as requiring a fixed top-to-bottom order and tolerating edited
cells (Section~\ref{sec:consistency-models}).

\steve{Note that \flowbook{} does not restrict out-of-order execution:
  users may still edit and run cells in any order.  However,
  \flowbook{} prevents that flexibility from silently producing
  outputs that depend on the order in which cells happened to run.  In
  contrast, standard practice offers a choice between interactivity
  with no consistency guarantee, or top-to-bottom execution after
  every edit.  The programming discipline that rerun consistency
  imposes is light.  A cell may still freely create and mutate new
  data structures.  Data structures from earlier cells cannot be
  mutated in place, but a cell may derive new values from them (e.g.,
  \texttt{df2 = df.drop(...)}).  The edit--rerun cycle that
  characterizes notebook development~\cite{TITOV2026112983} is
  unaffected.  Moreover, the costs of the two failure modes are
  asymmetric: \flowbook{} reports a violation immediately, and the
  user can fix it locally in the offending cell, whereas an
  inconsistent output looks valid and may go
  unnoticed.}

\flowbook{} interposes on the Python kernel's global namespace and
checkpoints state at cell boundaries to track read and write
sets. Efficient checkpointing exploits common properties of notebook
workloads, most notably that state often resides in pandas DataFrames,
whose copy-on-write semantics make snapshots
inexpensive~\cite{pandascow}.  \data{Across 61 real-world Kaggle
  competition notebooks, \flowbook{} flagged 203 rerun-inconsistent
  cells.  For at least 126 of them (62\%), rerunning the cell and
  every cell below it, or rerunning an earlier cell that reads a value
  it overwrote, changes a visible output, leaving the notebook output
  inconsistent.  That is, output inconsistency is common in practice.
  A further 21\% change notebook state without changing visible
  output; these cells may affect the results of later edits and
  reruns.  Another 11\% rewrite the same values when rerun, and only
  5\% are false positives of \flowbook{}'s conservative analysis.}  \data{All of these violations are easy to
  eliminate: \flowbook{}'s AI repair agent repaired all 203 cells with
  local edits and no human intervention.}  \data{\flowbook{}
  adds a median overhead of only 26\,ms per interactive cell
  execution, low enough to have little perceptible impact on
  interactive use.}

\steve{
In summary, this paper:
\begin{itemize}[topsep=0pt]
\item motivates our emphasis on output consistency (Section~\ref{sec:example});
 \item defines rerun consistency, presents a dynamic analysis to enforce
 it, and proves the analysis correct, with a mechanization of the proof in
 Lean~4~\cite{DBLP:conf/cade/Moura021}
 (Section~\ref{sec:analysis} and \ifarxiv Appendices~\ref{sec:semantics}--\ref{sec:lean}\else the Supplemental Material\fi);
\item describes \flowbook, our implementation for
  Python notebooks (Section~\ref{sec:implementation});
\item \data{demonstrates that rerun consistency is violated frequently in the
  real-world notebooks we examined, and that these violations often visibly
  change notebook output (Section~\ref{sec:evaluation});}
\item \data{shows that  \flowbook{}'s 
  AI repair agent automatically fixes these violations (Section~\ref{sec:evaluation}); and}
\item \data{shows that \flowbook{} median overhead is only 26\,ms per cell execution
  (Section~\ref{sec:evaluation}).}
\end{itemize}
\ifarxiv Appendices~\ref{sec:semantics}--\ref{sec:lean} include\else The Supplemental Material includes\fi{} \flowbook{}'s full formal development and correctness proof.
}

%% file: nbs/calibration-litmus-tests.tex
\begin{figure*}[tp!]
\setlength{\tabcolsep}{4pt}
\setlength{\nbAlignWidth}{\dimexpr0.32\textwidth-8.8pt\relax}

\def\litTitleA{\textbf{(1) Inconsistent output}: Cell \cell{D}'s plot does not reflect the calibration in \cell{C} because \cell{D} ran before \cell{C}.}
\def\litTitleB{\textbf{(2) Consistent output}: Outputs are consistent with the top-to-bottom execution.}
\def\litTitleC{\textbf{(3) Inconsistent output}: \texttt{df["pm25"]} was multiplied by both 1.8 and 1.85, so no rerun of \cell{D} alone plots the calibrated readings.}

\setlength{\nbHtA}{0pt}%
\begin{nbmeasure}{\nbHtA}\nbtitleblock{\litTitleA}\end{nbmeasure}
\begin{nbmeasure}{\nbHtA}\nbtitleblock{\litTitleB}\end{nbmeasure}
\begin{nbmeasure}{\nbHtA}\nbtitleblock{\litTitleC}\end{nbmeasure}

\setlength{\nbHtB}{0pt}%
\begin{nbmeasure}{\nbHtB}
\begin{nbcode}{4}{@C}
\begin{lstlisting}
df["pm25"] =
  df["pm25"]*1.8
\end{lstlisting}
\end{nbcode}
\end{nbmeasure}
\begin{nbmeasure}{\nbHtB}
\begin{nbedit}{3}{5}{@C}
\begin{lstlisting}
|df["pm25"] =|
  |df["pm25"]*1.8|
\end{lstlisting}
\nbeditsep
\begin{lstlisting}
df["pm25"] =
  df["pm25"]*1.85
\end{lstlisting}
\end{nbedit}
\end{nbmeasure}

\begin{center}
\Large Jupyter Notebooks Do Not Ensure Output Consistency
\end{center}\vspace{-2ex}
\begin{tabular}{ccc}

\begin{minipage}[t]{0.32\textwidth}\vspace{0pt}
\begin{nbnotebook}{}
\begin{nbslot}{\nbHtA}\nbtitleblock{\litTitleA}\end{nbslot}\vspace{5pt}
\nbflow{nbcA}
\begin{nbcode}{1}{@A}
\begin{lstlisting}
import pandas as pd
from pm25lib import plot
df = pd.read_csv(
     "data.csv")
\end{lstlisting}
\end{nbcode}

\begin{nbcode}{2}{@B}
\begin{lstlisting}
df["pm25"].mean()
\end{lstlisting}
\end{nbcode}
\begin{nboutput}{2}
11.2
\end{nboutput}

\begin{nbslot}{\nbHtB}
\begin{nbcode}{4}{@C}
\begin{lstlisting}
df["pm25"] =
  df["pm25"]*1.8
\end{lstlisting}
\end{nbcode}
\end{nbslot}

\begin{nbcode}{3}{@D}
\begin{lstlisting}
plot(df, "day", 
     "pm25")
\end{lstlisting}
\end{nbcode}
\begin{nboutput}{3}
\includegraphics[width=1\linewidth]{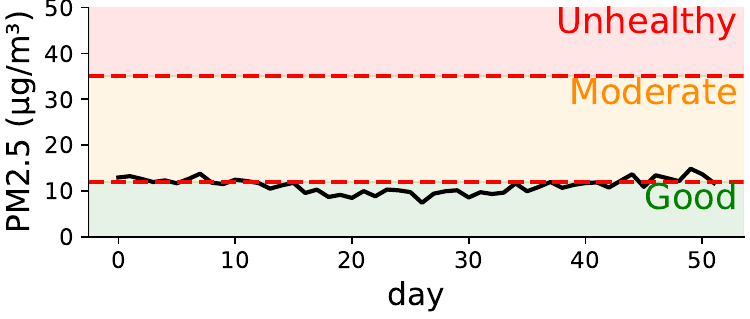}
\end{nboutput}
\end{nbnotebook}
\nbflowarrow[arc=2mm]{nbcA}{1}{2}
\nbflowarrow[arc=4mm, entry=south west]{nbcA}{2}{3}
\nbflowarrow[arc=2mm, exit=north west, entry=west]{nbcA}{3}{4}
\end{minipage}
&
\begin{minipage}[t]{0.32\textwidth}\vspace{0pt}
\begin{nbnotebook}{}
\begin{nbslot}{\nbHtA}\nbtitleblock{\litTitleB}\end{nbslot}\vspace{5pt}
\nbflow{nbcB}
\begin{nbcode}{1}{@A}
\begin{lstlisting}
import pandas as pd
from pm25lib import plot
df = pd.read_csv(
     "data.csv")
\end{lstlisting}
\end{nbcode}

\begin{nbcode}{2}{@B}
\begin{lstlisting}
df["pm25"].mean()
\end{lstlisting}
\end{nbcode}
\begin{nboutput}{2}
11.2
\end{nboutput}

\begin{nbslot}{\nbHtB}
\begin{nbcode}{3}{@C}
\begin{lstlisting}
df["pm25"] =
  df["pm25"]*1.8
\end{lstlisting}
\end{nbcode}
\end{nbslot}

\begin{nbcode}{4}{@D}
\begin{lstlisting}
plot(df, "day", 
     "pm25")
\end{lstlisting}
\end{nbcode}
\begin{nboutput}{4}
\includegraphics[width=1\linewidth]{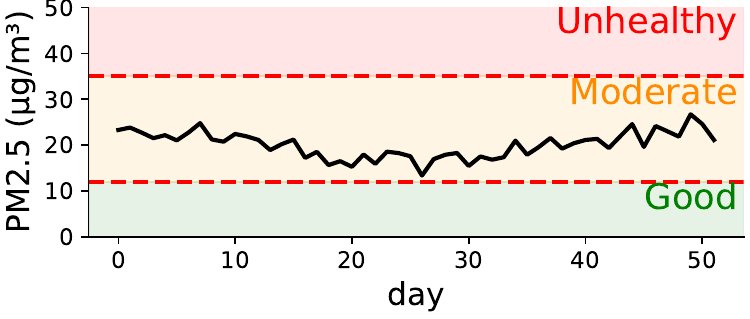}
\end{nboutput}
\end{nbnotebook}
\nbflowarrows{nbcB}{1,2,3,4}
\end{minipage}
&
\begin{minipage}[t]{0.32\textwidth}\vspace{0pt}
\begin{nbnotebook}{}
\begin{nbslot}{\nbHtA}\nbtitleblock{\litTitleC}\end{nbslot}\vspace{5pt}
\nbflow{nbcC}
\begin{nbcode}{1}{@A}
\begin{lstlisting}
import pandas as pd
from pm25lib import plot
df = pd.read_csv(
     "data.csv")
\end{lstlisting}
\end{nbcode}

\begin{nbcode}{2}{@B}
\begin{lstlisting}
df["pm25"].mean()
\end{lstlisting}
\end{nbcode}
\begin{nboutput}{2}
11.2
\end{nboutput}

\begin{nbslot}{\nbHtB}
\begin{nbedit}{3}{5}{@C}
\begin{lstlisting}
|df["pm25"] =|
  |df["pm25"]*1.8|
\end{lstlisting}
\nbeditsep
\begin{lstlisting}
df["pm25"] =
  df["pm25"]*1.85
\end{lstlisting}
\end{nbedit}
\end{nbslot}

\begin{nbcodererun}{4}{6}{@D}
\begin{lstlisting}
plot(df, "day", 
     "pm25")
\end{lstlisting}
\end{nbcodererun}
\begin{nboutput}{4}
\includegraphics[width=1\linewidth]{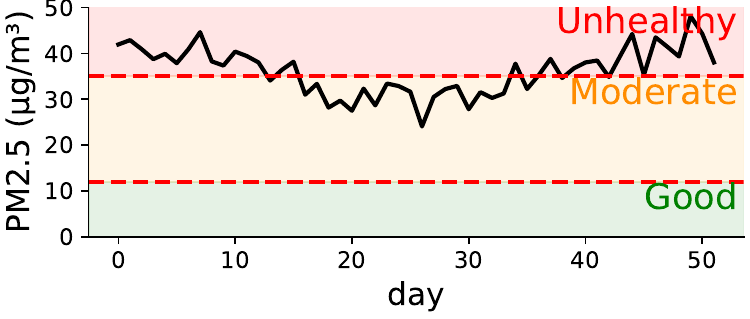}
\end{nboutput}
\end{nbnotebook}
\nbflowarrow[arc=2mm]{nbcC}{1}{2}
\nbflowarrow[arc=2mm]{nbcC}{2}{3}
\nbflowarrow[arc=3mm, entry=south west]{nbcC}{3}{4}
\nbflowarrow[arc=0.5mm, exit=north west, entry=south west]{nbcC}{4}{5}
\nbflowarrow[arc=-2mm, exit=east, entry=east]{nbcC}{5}{6}
\end{minipage}
\end{tabular}

\caption{Three scenarios illustrating common output consistency
  problems in a Jupyter notebook, described in detail in the text.  Each
  cell carries a unique label (e.g., \cell{A}).  The bracketed numbers
  (e.g., \texttt{[6]}) and arrows indicate execution order.  Edited cells show the original code
  with strikethrough, followed by the new code below the dashed line.
}
\label{fig:calibration-litmus-tests}
\end{figure*}

%% file: example.tex
\section{Motivating Example}
\label{sec:example}

\input{nbs/calibration-flowbook.tex}

Section~\ref{sec:introduction} demonstrates how Jupyter notebooks can produce outputs that no top-to-bottom execution would produce. Figure~\ref{fig:calibration-flowbook} shows how \flowbook{} ensures output consistency while allowing the scientist to edit and rerun cells interactively.

\myparagraph{(1) Rerun consistency violation}
As described in Section~\ref{sec:introduction}, a scientist discovers that the
sensor's raw readings were not properly calibrated after running cells \cell{A},
\cell{B}, and \cell{D}, and inserts a new cell \cell{C} to apply a calibration
gain. Cell \cell{C} overwrites \texttt{df["pm25"]}, which \cell{B} above reads.
Thus, a later rerun of \cell{B} would observe data that a top-to-bottom execution
could not. Similarly, cell \cell{C} reads and writes \texttt{df["pm25"]}, so
every re-execution re-applies the gain, compounding it into a value that no
top-to-bottom execution could produce. \flowbook{} detects and reports these
violations when \cell{C} is run.  The blocked run leaves the notebook state unchanged, so the scientist
continues from a state that a top-to-bottom execution could produce.

\myparagraph{(2) Rerun consistency restored}
Guided by the reported violation, the scientist writes the calibrated values to a fresh
column, modifies \cell{D} to plot that column, and reruns \cell{C}, which is now
rerun consistent.  Because \cell{D} was edited but not rerun, \flowbook{} marks
it as stale since the output it displays may no longer be generated by the
current code.

\myparagraph{(3) Propagating staleness}
When the scientist edits \cell{C} to use a new gain of 1.85 and reruns it,
\flowbook{} marks \cell{D} as stale because it reads \texttt{df["pm25\_cal"]},
which the edited cell \cell{C} wrote.  Rerunning \cell{D} restores output
consistency and shows a plot computed with the new gain.

\myparagraph{Key Takeaways}
This example illustrates how editing cells during iterative analysis breaks
output consistency.  \flowbook{} detects both rerun consistency violations (\cell{C} in the first panel) and
stale cells (\cell{D} in the second and third panels) and tells the scientist exactly which cells are clean and
which must be rerun.  Further, \flowbook{} does not require a top-to-bottom execution after
every edit or rerun consistency violation, preserving interactivity.  \steve{Throughout, the scientist edits and
runs cells in whatever order the analysis demands.  The only change to the
scientist's code is that \cell{C} derives a new column from \texttt{df["pm25"]}
rather than overwriting that column, and \flowbook{} reports the need for that
change the moment \cell{C} first runs, rather than leaving an inconsistent plot
that looks exactly like a valid one.}

%% file: nbs/calibration-flowbook.tex
\begin{figure*}[hbtp!]
\setlength{\tabcolsep}{4pt}
\setlength{\nbAlignWidth}{\dimexpr0.32\textwidth-8.8pt\relax}

\def\fbTitleA{\textbf{(1) Rerun consistency violation}: Cell \cell{C} modifies a location previously read by @B and itself.}
\def\fbTitleB{\textbf{(2) Rerun consistency restored}: Cell \cell{C} writes to a new column. Cell \cell{D} is stale because its code was edited.}
\def\fbTitleC{\textbf{(3) Propagating staleness}: Cell \cell{D} is stale because it reads a location modified by \cell{C}'s last run.}

\setlength{\nbHtA}{0pt}
\begin{nbmeasure}{\nbHtA}\nbtitleblock{\fbTitleA}\end{nbmeasure}
\begin{nbmeasure}{\nbHtA}\nbtitleblock{\fbTitleB}\end{nbmeasure}
\begin{nbmeasure}{\nbHtA}\nbtitleblock{\fbTitleC}\end{nbmeasure}

\setlength{\nbHtB}{0pt}
\begin{nbmeasure}{\nbHtB}
\begin{nbcodeerror}{X}{@C}
\begin{lstlisting}
df["pm25"] =
  df["pm25"]*1.8
\end{lstlisting}
\end{nbcodeerror}
\begin{nberror}
Writes \texttt{df["pm25"]} read by @B above and @C
\end{nberror}
\end{nbmeasure}
\begin{nbmeasure}{\nbHtB}
\begin{nbcode}{4}{@C}
\begin{lstlisting}
df["pm25_cal"] =
  df["pm25"]*1.8
\end{lstlisting}
\end{nbcode}
\end{nbmeasure}
\begin{nbmeasure}{\nbHtB}
\begin{nbedit}{4}{6}{@C}
\begin{lstlisting}
|df["pm25_cal"] =|
  |df["pm25"]*1.8|
\end{lstlisting}
\nbeditsep
\begin{lstlisting}
df["pm25_cal"] =
  df["pm25"]*1.85
\end{lstlisting}
\end{nbedit}
\end{nbmeasure}

\setlength{\nbHtC}{0pt}
\begin{nbmeasure}{\nbHtC}
\begin{nbcode}{3}{@D}
\begin{lstlisting}
plot(df, "day", "pm25")
\end{lstlisting}
\end{nbcode}
\begin{nboutput}{3}
\includegraphics[width=1\linewidth]{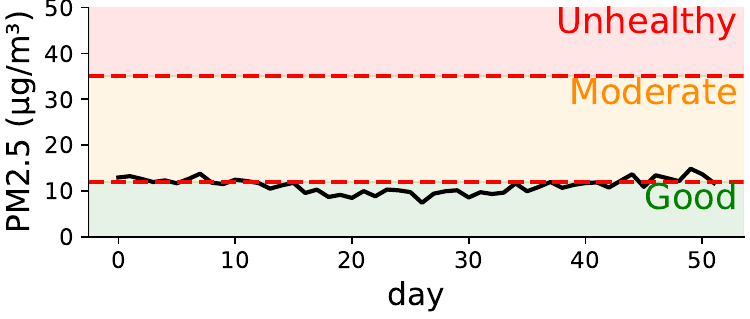}
\end{nboutput}
\end{nbmeasure}
\begin{nbmeasure}{\nbHtC}
\begin{nbeditwarn}{3}{X}{@D}{\cell{D} has been edited}
\begin{lstlisting}
|plot(df, "day", "pm25")|
\end{lstlisting}
\nbeditsep
\begin{lstlisting}
plot(df, "day",
     "pm25_cal")
\end{lstlisting}
\end{nbeditwarn}
\begin{nboutput}{4}
\includegraphics[width=1\linewidth]{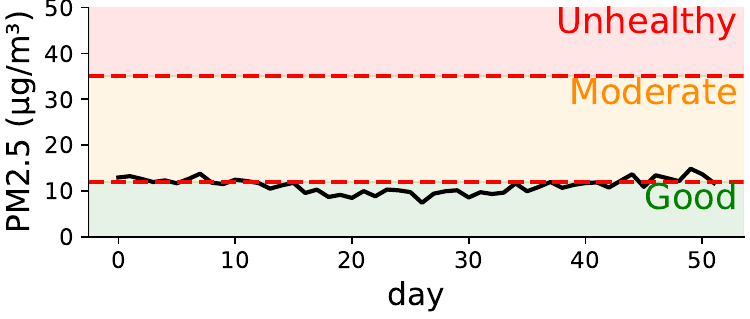}
\end{nboutput}
\end{nbmeasure}
\begin{nbmeasure}{\nbHtC}
\begin{nbeditwarn}{3}{5}{@D}{\texttt{df["pm25\_cal"]}\\\mbox{~~~~~~} modified by \cell{C}}
\begin{lstlisting}
|plot(df, "day", "pm25")|
\end{lstlisting}
\nbeditsep
\begin{lstlisting}
plot(df, "day",
     "pm25_cal")
\end{lstlisting}
\end{nbeditwarn}
\begin{nboutput}{4}
\includegraphics[width=1\linewidth]{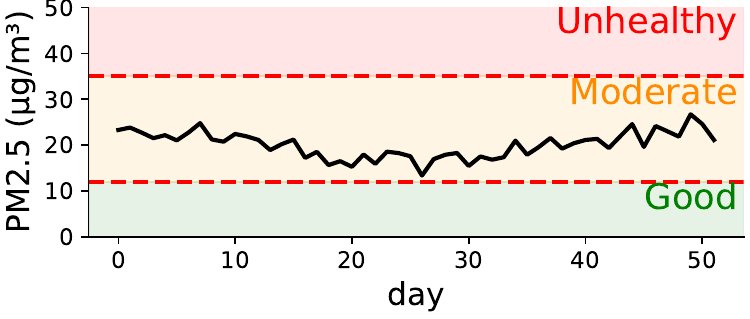}
\end{nboutput}
\end{nbmeasure}

\begin{center}
\Large \flowbook{} Ensures Output Consistency
\end{center}\vspace{-2ex}
\begin{tabular}{ccc}

\begin{minipage}[t]{0.32\textwidth}\vspace{0pt}
\begin{nbnotebook}{}
\begin{nbslot}{\nbHtA}\nbtitleblock{\fbTitleA}\end{nbslot}\vspace{5pt}
\nbflow{nbfbA}
\begin{nbcode}{1}{@A}
\begin{lstlisting}
import pandas as pd
from pm25lib import plot
df = pd.read_csv(
     "data.csv")
\end{lstlisting}
\end{nbcode}

\begin{nbcode}{2}{@B}
\begin{lstlisting}
df["pm25"].mean()
\end{lstlisting}
\end{nbcode}
\begin{nboutput}{2}
11.2
\end{nboutput}

\begin{nbslot}{\nbHtB}
\begin{nbcodeerror}{X}{@C}
\begin{lstlisting}
df["pm25"] =
  df["pm25"]*1.8
\end{lstlisting}
\end{nbcodeerror}
\begin{nberror}
Writes \texttt{df["pm25"]} read by @B above and @C
\end{nberror}
\end{nbslot}

\begin{nbslot}{\nbHtC}
\begin{nbcode}{3}{@D}
\begin{lstlisting}
plot(df, "day", 
     "pm25")
\end{lstlisting}
\end{nbcode}
\begin{nboutput}{3}
\includegraphics[width=1\linewidth]{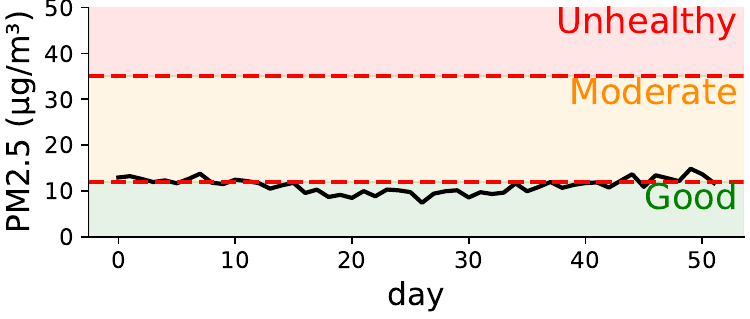}
\end{nboutput}
\end{nbslot}
\end{nbnotebook}
\nbflowarrow[arc=2mm]{nbfbA}{1}{2}
\nbflowarrow[arc=4mm, entry=south west]{nbfbA}{2}{3}
\nbflowarrow[arc=2mm, exit=north west, entry=west]{nbfbA}{3}{X}
\end{minipage}
&
\begin{minipage}[t]{0.32\textwidth}\vspace{0pt}
\begin{nbnotebook}{}
\begin{nbslot}{\nbHtA}\nbtitleblock{\fbTitleB}\end{nbslot}\vspace{5pt}
\nbflow{nbfbB}
\begin{nbcode}{1}{@A}
\begin{lstlisting}
import pandas as pd
from pm25lib import plot
df = pd.read_csv(
     "data.csv")
\end{lstlisting}
\end{nbcode}

\begin{nbcode}{2}{@B}
\begin{lstlisting}
df["pm25"].mean()
\end{lstlisting}
\end{nbcode}
\begin{nboutput}{2}
11.2
\end{nboutput}

\begin{nbslot}{\nbHtB}
\begin{nbcode}{4}{@C}
\begin{lstlisting}
df["pm25_cal"] =
  df["pm25"]*1.8
\end{lstlisting}
\end{nbcode}
\end{nbslot}

\begin{nbslot}{\nbHtC}
\begin{nbeditwarn}{3}{X}{@D}{\cell{D} has been edited}
\begin{lstlisting}
|plot(df, "day", "pm25")|
\end{lstlisting}
\nbeditsep
\begin{lstlisting}
plot(df, "day",
     "pm25_cal")
\end{lstlisting}
\end{nbeditwarn}
\begin{nboutput}{4}
\includegraphics[width=1\linewidth]{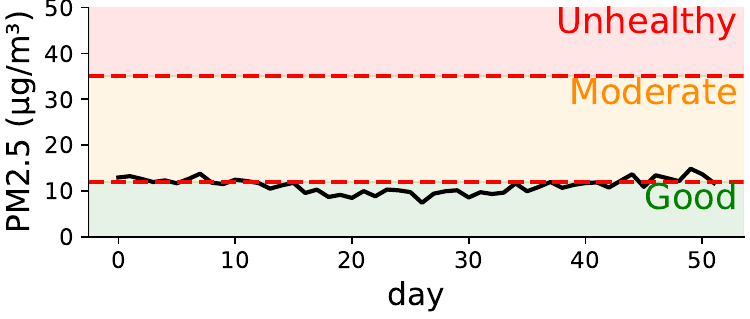}
\end{nboutput}
\end{nbslot}
\end{nbnotebook}
\nbflowarrow[]{nbfbB}{1}{2}
\nbflowarrow[entry=south west]{nbfbB}{2}{3}
\nbflowarrow[arc=2mm, exit=north west, entry=west]{nbfbB}{3}{4}
\end{minipage}
&
\begin{minipage}[t]{0.32\textwidth}\vspace{0pt}
\begin{nbnotebook}{}
\begin{nbslot}{\nbHtA}\nbtitleblock{\fbTitleC}\end{nbslot}\vspace{5pt}
\nbflow{nbfbC}
\begin{nbcode}{1}{@A}
\begin{lstlisting}
import pandas as pd
from pm25lib import 
  plot
df = pd.read_csv(
     "data.csv")
\end{lstlisting}
\end{nbcode}

\begin{nbcode}{2}{@B}
\begin{lstlisting}
df["pm25"].mean()
\end{lstlisting}
\end{nbcode}
\begin{nboutput}{2}
11.2
\end{nboutput}

\begin{nbslot}{\nbHtB}
\begin{nbedit}{4}{6}{@C}
\begin{lstlisting}
|df["pm25_cal"] =|
  |df["pm25"]*1.8|
\end{lstlisting}
\nbeditsep
\begin{lstlisting}
df["pm25_cal"] =
  df["pm25"]*1.85
\end{lstlisting}
\end{nbedit}
\end{nbslot}

\begin{nbslot}{\nbHtC}
\begin{nbeditwarn}{3}{5}{@D}{\texttt{df["pm25\_cal"]}\\\mbox{~~~~~~} modified by \cell{C}}
\begin{lstlisting}
|plot(df, "day", "pm25")|
\end{lstlisting}
\nbeditsep
\begin{lstlisting}
plot(df, "day",
     "pm25_cal")
\end{lstlisting}
\end{nbeditwarn}
\begin{nboutput}{4}
\includegraphics[width=1\linewidth]{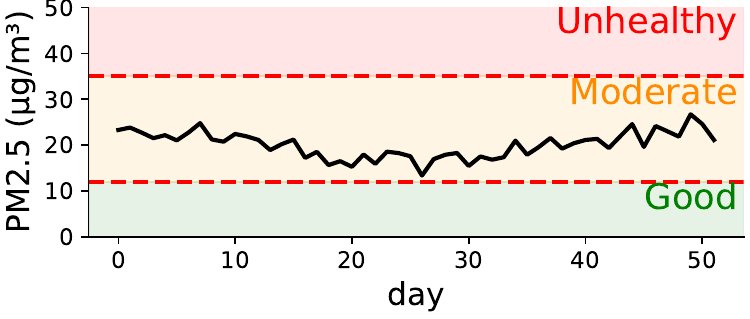}
\end{nboutput}
\end{nbslot}
\end{nbnotebook}
\nbflowarrow[arc=2mm]{nbfbC}{1}{2}
\nbflowarrow[arc=8mm, entry=south west]{nbfbC}{2}{3}
\nbflowarrow[arc=3mm, entry=south west]{nbfbC}{3}{4}
\nbflowarrow[arc=-3mm, exit=south east, entry=south east]{nbfbC}{4}{5}
\nbflowarrow[arc=-1mm, exit=north east, entry=south east]{nbfbC}{5}{6}
\end{minipage}

\end{tabular}
\caption{\flowbook{} on the air-quality notebook of
  Figure~\ref{fig:calibration-litmus-tests}, using the same cells and
  conventions.  A cell that has not yet run successfully has an empty execution count (\textcolor{nbprompt}{\texttt{\small\bfseries[\phantom{0}]}}).
  In contrast to standard Jupyter notebooks (\autoref{fig:calibration-litmus-tests}), \flowbook{} ensures output consistency by enforcing rerun consistency, blocking cells that write to variables read in an earlier cell (left) and 
  warning about staleness (middle, right).
}
\label{fig:calibration-flowbook}
\end{figure*}

%% file: flowbook-analysis.tex
\input{nbs/litmus-tests.tex}

\section{FlowBook Analysis}
\label{sec:analysis}

\flowbook{} enforces output consistency by observing every cell as it
runs. For each executed cell, \flowbook{} records the set of locations
the cell reads, the set of locations it writes, and a status tag marking
the cell as either \clean{} or \stale{}.  Using that information, it rejects any run that a top-to-bottom execution
could not have produced, and it propagates staleness to the cells whose
recorded outputs may no longer be valid. This section
presents the analysis and its guarantees.

\subsection{Notebook Model and Interactive Execution}

A notebook is a sequence of cells. Each cell has source code and a most
recent output; a cell that has never executed has no output. Cells do
not pass values to one another directly. Instead, they communicate
through a shared \emph{store} that maps locations to values, exactly as
a Python kernel's global namespace maps names to objects. Running a cell
reads some locations from the store, writes others, and replaces the
cell's output.

The granularity at which \flowbook{} tracks locations is central to its
precision. A \emph{location} in \flowbook{} is either a top-level variable~$x$ or an
individual column $d.c$ of the DataFrame at address~$d$:
\[
\ell \in \Loc ::= x \mid d.c
\]
This choice
reflects a key observation about how data-analysis notebooks are
written, namely that a DataFrame column is the natural unit of independent
modification. Tracking
columns separately lets \flowbook{} recognize that two cells touching
disjoint columns of the same DataFrame do not interfere. \flowbook{} treats every other
structure---lists, dictionaries, NumPy arrays, scalars---as an atomic
value stored at a single location: a read or write of any part is a read
or write of the whole.
\steve{This treatment is sound but less precise. For example, two cells that touch
disjoint keys of a dictionary or disjoint slices of an array are
treated as interfering. Tracking dictionary keys or array slices as
separate locations, as \flowbook{} does for DataFrame columns, is 
natural  future work.}

Users build notebooks incrementally by inserting, deleting, editing, moving, and running cells in whatever order the task demands. The defining feature
of this interactive style is that a user may run \emph{any} cell at
\emph{any} time.

\input{nbs/staleness-litmus-tests.tex}

\subsection{Output Consistency}
\label{sec:output-consistency}

A notebook is output consistent if its recorded outputs match those that could be produced by a top-to-bottom execution from an empty store. Because cells may be nondeterministic (e.g., by drawing random numbers), there may be many such executions. We require only that the outputs agree with one valid top-to-bottom execution.

\begin{definition}[Output-Consistent Notebook]
  \label{def:output-consistent}
  A notebook is \emph{output consistent} if and only if there exists a
  top-to-bottom execution of all its cells from an empty store whose
  outputs exactly match those currently recorded in the notebook.
\end{definition}

This top-to-bottom execution, the \emph{reference execution}, need not reproduce the user's intermediate \emph{store}; only the visible \emph{outputs} must agree. The user's store reflects the particular interactive history that produced the notebook, whereas the reference execution may follow a different sequence of intermediate states. \flowbook{} guarantees that whenever it reports every cell to be \clean{}, the notebook is output consistent in this sense.

\subsection{Tracking Reads, Writes, and Status}

\flowbook{} maintains an instrumentation state $(R, W, T)$ where for each cell $i$:
\begin{itemize}
\item $R_i \subseteq \Loc$ is the set of locations read (before written),
\item $W_i \subseteq \Loc$ is the set of locations written,
\item $T_i \in \{\clean, \stale\}$ is the status tag.
\end{itemize}
A \clean{} tag asserts that the cell's recorded output is still valid with respect to the current store; a \stale{} tag indicates that the cell's recorded output may be invalid and the cell must be re-executed. Every cell starts with $T_i = \stale{}$, $R_i = \emptyset$, and $W_i = \emptyset$. 

When a cell runs, \flowbook{} instruments the underlying Python runtime
to capture which locations it actually reads from and writes to the
store. The cell
\[
\texttt{x = 0; \ x = x + 1}
\]
never observes a prior value of \texttt{x}: the second statement reads
the \texttt{0} just written in the same cell, which is not part of the
incoming store, so this cell reads nothing and writes only \texttt{x}.

\subsection{Rerun Consistency}
\label{sec:rerun-consistency}

\flowbook permits a cell's execution only if the cell's reads and writes
are consistent with what a top-to-bottom execution would have done, that
is, only if executing the cell \emph{in place} at its current position
performs the same reads and writes it would in a top-to-bottom
execution. We call this \emph{rerun consistency}.

\begin{definition}[Rerun-Consistent Cell] \label{def:rerun-consistent-accesses}
  Instrumentation state $R$ and $W$ are \emph{rerun consistent} for cell $i$
  if these predicates hold:
  \begin{itemize}
    \item $\NoReadAndWrite(R, W, i)$: cell $i$ does not read and write the same location.
    \item $\WriteBeforeRead(R, W, i)$: every location read by $i$ was written by a cell above $i$.
    \item $\NoReadBeforeWrite(R, W, i)$: cell $i$ does not read locations written by cells below $i$.
    \item $\NoWriteAfterRead(R, W, i)$: cell $i$ does not overwrite locations read by cells above $i$.
  \end{itemize}
\end{definition}
We define these predicates as follows.
\begin{eqnarray*}
  \NoReadAndWrite(R, W, i) &\defeq & R_i \cap W_i = \emptyset \\
  \WriteBeforeRead(R, W, i) &\defeq & R_i \subseteq (\cup\, W_{1..i-1}) \\
  \NoReadBeforeWrite(R, W, i) &\defeq & R_i \cap (\cup\, W_{i+1..n}) = \emptyset \\
  \NoWriteAfterRead(R, W, i) &\defeq &  W_i \cap (\cup\, R_{1..i-1}) = \emptyset
\end{eqnarray*}
If a cell run violates any of these, \flowbook{} blocks it
and reports the violation (Figure~\ref{fig:litmus-tests}).

\subsection{Propagating Staleness}
\label{sec:staleness}

After the user runs or deletes a cell, \flowbook{}
marks as \stale{} every cell that may no longer be rerun consistent with the current store.  There are two cases: forward and backward staleness.

\myparagraph{Forward staleness}
Running cell~$i$ may overwrite a location that some cell below~$i$ reads
or writes. Any such cell's recorded output may no longer be valid with respect to
the store, so the cell is marked \stale{}. \flowbook{} also accounts for locations the cell
\emph{stopped} writing: if a rerun of cell~$i$ no longer writes a
location it previously wrote, a cell below that read that
location is likewise marked \stale{}. Concretely, a cell~$j$ below~$i$ becomes
stale whenever its reads or writes intersect the locations cell~$i$
writes now or wrote before
($j > i$ and $(W_i \cup W_i') \cap (R_j \cup W_j) \neq \emptyset$, where
$W_i'$ is the new write set). The same rule applies after a cell is
deleted, using the deleted cell's former write set 
(Figure~\ref{fig:staleness-litmus-tests}).

\myparagraph{Backward staleness}
Suppose cell~$i$ is edited and
rerun so that it no longer writes a location~$\ell$ it used to write.
Cells below~$i$ that read~$\ell$ were relying on cell~$i$ to supply it;
with that write gone, the value must instead come from the nearest cell
\emph{above}~$i$ that writes~$\ell$. \flowbook{} marks that cell, the
last writer of~$\ell$ above~$i$, as stale, so that rerunning it
restores~$\ell$ to a value consistent with a top-to-bottom execution.

\myparagraph{Other interactions}
Interactions beyond running or deleting cells do not require
additional staleness propagation: editing a cell simply marks
it \stale{}; insertion adds a new \stale{} cell; and moving a cell
behaves as a deletion followed by an insertion, inheriting their
staleness behavior.

\subsection{Guarantees}
\label{sec:guarantees}

The rerun consistency and staleness propagation mechanisms
maintain the following invariant:

\begin{definition}[Well-Formed Notebook]
  \label{def:well-formed}
  A notebook is \emph{well-formed} if for every cell $i$ marked \clean{}
  we have that:
  \begin{enumerate}
      \item $R$ and $W$ are rerun consistent for $i$, and
      \item cell $i$ can be
  re-executed from the current store to reproduce 
  the recorded output, read set $R_i$, and write set $W_i$,
  where each write either (1) does not change the store or (2) only changes a location that is overwritten by a write below $i$.\footnote{For nondeterministic cells, there must be at least one possible execution satisfying these requirements.}
  \end{enumerate}
\end{definition}
\noindent
That is, a clean cell does not violate rerun consistency, and its re-execution preserves the store (modulo later writes), the recorded outputs, and the instrumentation state.

Three properties follow, which we state informally here and prove as theorems in
 \ifarxiv Appendix~\ref{sec:correctness}\else the Supplemental Material\fi. \steve{We also mechanized our proofs in Lean~\cite{DBLP:conf/cade/Moura021}.}

\ifarxiv\else\clearpage\fi
\noindent\textbf{Preservation.} Every interaction preserves
  well-formedness.

\smallskip
\noindent\textbf{Output consistency.} If all cells are \clean{} in a well-formed notebook, it is output consistent.

\smallskip
\noindent\textbf{Progress.}
\steve{In a well-formed notebook, repeatedly running the first stale cell
terminates in one of two states: (1) every cell is clean, and
hence, by the previous property, the notebook is output consistent; or (2)
execution halts at a stuck cell, which fails a rerun consistency check
or raises an exception. The user must resolve these issues before execution
can proceed. This property assumes that cells are deterministic: the
location a cell reads next, and the locations and values it writes,
depend only on the values it has read so far. Python executes a cell's
statements in order, so every deterministic cell satisfies this
assumption. (A set-based form of determinism, which ignores the order of
reads, does not suffice; see \ifarxiv Appendix~\ref{sec:correctness}\else the Supplemental Material\fi.)
On our benchmark suite, only five of the 999 cells fail to reproduce
their read and write sets when rerun with unchanged inputs, due to
either nondeterminism or sound over-approximations in how \flowbook{}
tracks and checkpoints DataFrames in several corner cases. For all five,
rerunning the first stale cell still reaches output consistency without
user intervention.}

\smallskip
Together, these properties yield \flowbook{}'s central guarantee:
starting from the initial all-stale state, any sequence of interactions
leaves the notebook well-formed, and any notebook in which \flowbook{} reports every cell to be
\clean{} is output consistent.

\subsection{Comparision to Other Consistency Models}
\label{sec:consistency-models}
\steve{
\flowbook's consistency guarantees relate to consistency models in databases and
distributed
systems~\cite{Terry2011ReplicatedData,eswaran1976notions,DBLP:journals/toplas/HerlihyW90,DBLP:journals/tc/Lamport79}.
Rerun consistency resembles consistent prefix semantics~\cite{Terry2011ReplicatedData} in requiring a cell to read values
produced by preceding cells in a designated order: the notebook's top-to-bottom
order. It additionally prevents later cells from changing values read by earlier
cells, so those earlier cells can be rerun consistently. \flowbook's consistency guarantees also
resemble serializability~\cite{eswaran1976notions} in comparing an interactive
execution with a serial reference execution, with cells playing the role of
transactions. The comparison differs in three respects. First, serializability
permits any serial order, whereas output consistency requires the notebook's
top-to-bottom order. Second, cells may be edited between executions, so the
reference execution uses their current code rather than necessarily the code
that produced their recorded outputs. Third, the reference begins from an empty
store and may never have occurred; serializability instead reorders the
transactions that actually executed.
}

%% file: nbs/litmus-tests.tex
\begin{figure}[htp!]
\def\litw{0.48}%
\def\litgap{\hfill}%
\def\litrow{\par\medskip}%
\begin{minipage}[t]{\litw\textwidth}\vspace{0pt}
\begin{nbnotebook}{1. \NoReadAndWrite}
\nbflow{nbAB}
\begin{nbcode}{1}{@A}
\begin{lstlisting}
x = 0
\end{lstlisting}
\end{nbcode}

\begin{nbcodeerror}{2}{@B}
\begin{lstlisting}
x = x + 1
\end{lstlisting}
\end{nbcodeerror}
\begin{nberror}
Reads and writes \texttt{x}
\end{nberror}
\end{nbnotebook}
\nbflowarrows{nbAB}{1,2}
\end{minipage}
\litgap\begin{minipage}[t]{\litw\textwidth}\vspace{0pt}
\begin{nbnotebook}{2. \WriteBeforeRead}
\begin{nbcode}{~}{@C}
\begin{lstlisting}
df = pd.read_csv("f.csv")
\end{lstlisting}
\end{nbcode}

\begin{nbcodeerror}{1}{@D}
\begin{lstlisting}
df.head()
\end{lstlisting}
\end{nbcodeerror}
\begin{nberror}
\texttt{df} not written by cell above
\end{nberror}
\end{nbnotebook}
\end{minipage}
\litrow\begin{minipage}[t]{\litw\textwidth}\vspace{0pt}
\begin{nbnotebook}{3. \NoReadBeforeWrite}
\nbflow{nbEF}
\begin{nbcodeerror}{2}{@E}
\begin{lstlisting}
df["y"].sum()
\end{lstlisting}
\end{nbcodeerror}
\begin{nberror}
Reads \texttt{df["y"]} written by \cell{F} below
\end{nberror}

\begin{nbcode}{1}{@F}
\begin{lstlisting}
df["y"] = [1, 2, 3]
\end{lstlisting}
\end{nbcode}
\end{nbnotebook}
\nbflowarrows{nbEF}{1,2}
\end{minipage}
\litgap\begin{minipage}[t]{\litw\textwidth}\vspace{0pt}
\begin{nbnotebook}{4. \NoWriteAfterRead}
\nbflow{nbGHI}
\begin{nbcode}{1}{@G}
\begin{lstlisting}
a = [1,2,3]
\end{lstlisting}
\end{nbcode}

\begin{nbcode}{2}{@H}
\begin{lstlisting}
sum(a)
\end{lstlisting}
\end{nbcode}

\begin{nbcodeerror}{3}{@I}
\begin{lstlisting}
a = 100
\end{lstlisting}
\end{nbcodeerror}
\begin{nberror}
Writes \texttt{a} read by \cell{H} above
\end{nberror}
\end{nbnotebook}
\nbflowarrows{nbGHI}{1,2,3}
\end{minipage}
\caption{Scenarios violating each predicate from Definition~\ref{def:rerun-consistent-accesses}:
(1)~\NoReadAndWrite{} fails when cell \cell{B} both reads and writes \texttt{x};
(2)~\WriteBeforeRead{} fails when cell \cell{D} reads \texttt{df} with no cell above \cell{D} writing to it (cell \cell{C} has not been run);
(3)~\NoReadBeforeWrite{} fails when cell 
    \cell{E} reads \texttt{df["y"]} written by a cell below \cell{E};
(4)~\NoWriteAfterRead{} fails when cell \cell{I} writes \texttt{a} previously read by cell \cell{H}.}
\label{fig:litmus-tests}
\end{figure}

%% file: nbs/staleness-litmus-tests.tex
\begin{figure}
\centering
\begin{minipage}[t]{0.48\linewidth}\vspace{0pt}
\input{nbs/stale-sc1-col.tex}
\end{minipage}\hfill
\begin{minipage}[t]{0.48\linewidth}\vspace{0pt}
\input{nbs/stale-sc2-col.tex}
\end{minipage}
\par\medskip
\begin{minipage}[t]{0.48\linewidth}\vspace{0pt}
\input{nbs/stale-sc3-col.tex}
\end{minipage}\hfill
\begin{minipage}[t]{0.48\linewidth}\vspace{0pt}
\input{nbs/stale-sc4-col.tex}
\end{minipage}
\caption{Scenarios in which \flowbook{} marks cells stale:
(1) rerunning an edited \cell{J} marks \cell{K} stale because \cell{J} writes \texttt{x};
(2) running \cell{L} after \cell{M} causes \ForwardStale{} to mark \cell{M} stale because \cell{L} writes \texttt{y};
(3) deleting \cell{N} after running \cell{N} and \cell{O} causes \cell{O} to become stale because the source of \texttt{df} no longer exists;
(4) after \cell{Q} is edited to stop writing \texttt{z} and rerun, \BackwardStale{} marks \cell{P} stale because it must restore \texttt{z}, and \ForwardStale{} marks \cell{R} stale.}
\label{fig:staleness-litmus-tests}
\end{figure}

%% file: nbs/stale-sc1-col.tex
\begin{nbnotebook}{1. \ForwardStale{} (write$\to$read)}
\nbflow{scJKcol}
\begin{nbedit}{1}{3}{@J}
\begin{lstlisting}
|x = 100|
\end{lstlisting}
\nbeditsep
\begin{lstlisting}
x = 9999
\end{lstlisting}
\end{nbedit}

\begin{nbcodewarn}{2}{@K}
\begin{lstlisting}
print(x)
\end{lstlisting}
\end{nbcodewarn}
\begin{nbwarning}[\linewidth]
\texttt{x} modified by \cell{J} above
\end{nbwarning}
\end{nbnotebook}
\nbflowarrow[arc=6mm,entry=south west]{scJKcol}{1}{2}
\nbflowarrow[arc=4mm,exit=north west, entry=west]{scJKcol}{2}{3}

%% file: nbs/stale-sc2-col.tex
\begin{nbnotebook}{2. \ForwardStale{} (write$\to$write)}
\nbflow{scMLcol}
\begin{nbcode}{2}{@L}
\begin{lstlisting}
y = 10
\end{lstlisting}
\end{nbcode}

\begin{nbcodewarn}{1}{@M}
\begin{lstlisting}
y = 20
\end{lstlisting}
\end{nbcodewarn}
\begin{nbwarning}
\texttt{y} modified by \cell{L} above
\end{nbwarning}
\end{nbnotebook}
\nbflowarrows{scMLcol}{1,2}

%% file: nbs/stale-sc3-col.tex
\begin{nbnotebook}{3. \ForwardStale{} (delete)}
\nbflow{scNOcol}
\begin{nbghost}{1}{@N}
\begin{lstlisting}
df = ...
\end{lstlisting}
\end{nbghost}

\begin{nbcodewarn}{2}{@O}
\begin{lstlisting}
df.describe()
\end{lstlisting}
\end{nbcodewarn}
\begin{nbwarning}[\linewidth]
\texttt{df} written by a deleted cell
\end{nbwarning}
\end{nbnotebook}
\nbflowarrows{scNOcol}{1,2}

%% file: nbs/stale-sc4-col.tex
\begin{nbnotebook}{4. \BackwardStale{}}%
\nbflow{scPQRcol}
\begin{nbcodewarn}{1}{@P}
\begin{lstlisting}
z = 0
\end{lstlisting}
\end{nbcodewarn}
\begin{nbwarning}[\linewidth]
\texttt{z} write conflict with \cell{Q}
\end{nbwarning}

\begin{nbedit}{2}{4}{@Q}
\begin{lstlisting}
|z = 99|
\end{lstlisting}
\nbeditsep
\begin{lstlisting}
other = 99
\end{lstlisting}
\end{nbedit}

\begin{nbcodewarn}{3}{@R}
\begin{lstlisting}
print(z)
\end{lstlisting}
\end{nbcodewarn}
\begin{nbwarning}[\linewidth]
\texttt{z} modified by \cell{Q} above
\end{nbwarning}
\end{nbnotebook}
\nbflowarrow[arc=4mm]{scPQRcol}{1}{2}
\nbflowarrow[arc=6mm,entry=south west]{scPQRcol}{2}{3}
\nbflowarrow[arc=4mm,exit=north west, entry=west]{scPQRcol}{3}{4}

%% file: implementation.tex
\section{Implementation}
\label{sec:implementation}

\begin{figure}[t]
  \input{nbs/mutation-example}
  \caption{Element-wise DataFrame updates are not permitted.
    (a) Editing and rerunning \cell{C} to update a different element
    leaves the original element-wise update in effect, causing
    inconsistent output.
    (b) \flowbook{} rejects the element-wise update and requires
    full-column assignment.}
  \label{fig:mutation-example}
\end{figure}

We implement \flowbook{} as a Jupyter plugin.  Its core analysis is
embedded in an extension to \<ipykernel>, the standard Python backend
for Jupyter notebooks.

For cell execution, \flowbook{} runs the cell's code with sufficient
instrumentation to ensure that the cell preserves well-formedness and
does not violate rerun consistency.
For edit, insert, delete, and move operations, \flowbook{} intercepts
the corresponding Jupyter protocol messages and updates the
instrumentation state in the kernel without re-executing any cell.
\flowbook{} communicates cell status to the user through the standard Jupyter
front-end: clean cells display normally, stale cells display with a
yellow marker and a message identifying the cause, and a cell whose
run was blocked displays a red marker and a message describing the
violation.

Rerun consistency requires full-column
assignment (\<df["x"] = ...>) rather than element-wise updates
(\<df.loc[0, "x"] = 99>). The reason is illustrated in
Figure~\ref{fig:mutation-example}: if cell \cell{C} updates a single
element and is later edited to update a different element, the original
update persists in the DataFrame, so the recorded output no longer
matches any top-to-bottom execution. Full-column assignment avoids this
problem because re-executing the cell restores the entire column to
its computed value.
This restriction applies to all element-wise operations,
including \<df.loc>, \<df.iloc>, \<df.at>, and methods with
\<inplace=True>. In practice, most data transformations naturally use
full-column operations; element-wise updates are typically confined to
initialization cells that can be rewritten in terms of full-column operations.

The remainder of this section outlines the most salient aspects of \flowbook{}'s implementation.

\subsection{Instrumented Cell Execution}
\label{sec:mutation}

To construct read and write sets, \flowbook{} interposes on the Python kernel's global namespace \texttt{user\_ns}, which realizes the store of Section~\ref{sec:analysis}, and on the pandas API to intercept accesses to DataFrames and Series. 
Of course, entries in \texttt{user\_ns} may point to other heap-allocated mutable data (e.g., arrays, lists, and other objects) where reads and writes via Python code, external libraries, or even native code are invisible to simple namespace tracking.
\flowbook{} conservatively assumes that a read of a global variable
\texttt{x} in \texttt{user\_ns} may also implicitly read any data structure
reachable from \texttt{x}, and so adding \texttt{x} to $R_i$ correctly
over-approximates all reads.  To identify writes, \flowbook{}
checkpoints the \texttt{user\_ns} before cell execution via an efficient
deep copy (using copy-on-write and other optimizations), and then diffs this checkpoint with the post-execution
\texttt{user\_ns} to identify all modified data structures.

To guarantee rerun consistency, \flowbook{} does not permit later
partial modification of data structures allocated in an earlier cell,
for the same reason that partial column updates are forbidden, as
previously illustrated in Figure~\ref{fig:mutation-example}. Instead,
rerun consistency requires that any data structure
already allocated by an earlier cell execution not be modified
by a later cell execution.  We call such a rerun consistency violation
an \emph{invalid mutation}; Table~\ref{tab:repro-errors} lists it among
the violation categories.  Our experimental results in
Section~\ref{sec:evaluation} show that invalid mutations are readily
repaired.

\flowbook{} can detect invalid mutations and other rerun consistency
violations only after executing the cell.  When it detects a violation,
\flowbook{} uses the checkpoint to roll back.
Thus, cell execution proceeds in three phases:
\begin{enumerate}
  \item \textbf{Pre-execution:} \flowbook{} checkpoints the
    namespace.
  \item \textbf{Execution:} \flowbook{} runs the cell under
    instrumentation to record reads and writes to variables and columns.
  \item \textbf{Post-execution:} (1) \flowbook{} diffs the checkpoint
    against the current namespace to detect mutations.  
    If any invalid mutations are discovered, the kernel
    rolls back and reports them.
    (2) \flowbook{} then checks the rerun consistency predicates with
    the observed $R_i$ and $W_i$.  If they hold, the kernel updates the instrumentation state and propagates staleness.
    Otherwise, the kernel rolls back to the checkpoint and reports a 
    violation.
\end{enumerate}

\steve{The guarantees of Section~\ref{sec:guarantees} assume that
  $R_i$ and $W_i$ cover every location that cell~$i$ accesses and that
  a blocked execution leaves the store unchanged. \flowbook{} meets
  these assumptions: a read of a global counts as a read of everything
  reachable from it, an over-approximation that can only cause more
  checks to fail; writes by diffing against the checkpoint, which
  catches mutations in library and native code; and it prohibits reads
  of values that cannot be checkpointed
  (Section~\ref{sec:checkpointing}). This deliberate conservatism 
  ensures soundness at low cost (Section~\ref{sec:eval-violations}).}

\subsection{DataFrame Attributes}
\label{sec:locations}

Data science workflows frequently inspect DataFrame metadata such as \texttt{shape}, \texttt{columns}, and \texttt{index}. Although these attributes are not writeable by the user and distinct from column locations, they reveal structural information that can participate in rerun consistency violations. For example, reading \texttt{df.columns} conflicts with later additions or removals of columns.
We therefore extend our model to treat DataFrame attributes as distinct locations and account for conflicts between attribute accesses and subsequent structural updates to the DataFrame.

\subsection{Creating Checkpoints}
\label{sec:checkpointing}

\flowbook{} creates checkpoints using specialized
deep-copying that exploits properties of typical notebooks and
data science libraries for efficiency.
Checkpoints preserve pointer aliases between
objects, so that if $x$ and $y$ are aliases, then
$x$ and $y$ refer to the same object in the
checkpoint.

\myparagraph{Non-deep-copyable objects and state}
Certain objects cannot be deep-copied.  These include generators, iterators, file handles, and matplotlib
figures. \flowbook{} excludes these from checkpoints and produces a
visible warning indicating that a variable could not be checkpointed.
To preserve analysis soundness, \flowbook{} prevents all subsequent reads 
of such a variable.  \data{This scenario is rare.  Across our benchmark suite
(Section~\ref{sec:evaluation}), only a single cell was impacted by this limitation.}  

\myparagraph{pandas copy-on-write semantics}
For pandas \texttt{DataFrame}s and \texttt{Series}, \flowbook{} creates shallow checkpoints that leverage pandas' copy-on-write semantics, sharing memory until data is modified. To avoid aliasing issues, columns containing mutable objects are copied eagerly.

\myparagraph{Immutable library structures}
\flowbook{} exploits properties of common machine learning libraries to reduce checkpointing overhead. Frameworks such as Keras~\cite{2018ascl.soft06022C} and PyTorch~\cite{DBLP:conf/nips/PaszkeGMLBCKLGA19} encapsulate a small set of mutable parameters within large immutable object graphs, allowing \flowbook{} to copy only mutable state while sharing immutable components; similarly, it avoids duplicating immutable trained structures exposed by libraries such as LightGBM~\cite{DBLP:conf/nips/KeMFWCMYL17}.

\myparagraph{GPU-aware copying} \flowbook{} usually does not need to track GPU memory when checkpointing because most GPU libraries either leave no user-accessible GPU state between cell executions or provide their own GPU-aware copying mechanisms. For GPU-resident RAPIDS cuDF DataFrames~\cite{rapids2026}, however, \flowbook{} checkpoints data directly on the GPU to avoid transfer overhead, falling back to the CPU when GPU memory is low.

\subsection{Checkpoint Comparison to Identify Mutation}

After cell execution, \flowbook{} diffs the namespace against the pre-execution checkpoint to identify invalid mutations. The diff algorithm combines structural and value-based comparison to handle aliased objects correctly, using specialized strategies for common notebook data types: column-level comparison for \texttt{DataFrame}s, element-wise comparison for \texttt{Series}, arrays, and other iterables, and recursive comparison of \texttt{\_\_dict\_\_} fields for general objects. To reduce overhead, \flowbook{} limits comparison to memory reachable from globals read or written by the cell, since any mutation must be reachable from one of these roots.

\subsection{File System Locations and State} Cells that read or write
  files introduce dependencies on external state that may lead to
  rerun consistency violations. \flowbook{} assumes any files read by
  the notebook were written by an implicit cell above all others, and
  it interposes on \texttt{open} and related I/O functions to record
  file paths in the cell's read or write set.  It also uses a
  virtual file system to checkpoint and roll back files.

\subsection{Supporting ``Diagnostic'' Cells}

Often, cells are used to inspect a data structure before a
later cell writes to it, for debugging, sanity checking, or
exploratory visualization.  To support this idiom, \flowbook{} allows
cells to be marked as diagnostic using the \texttt{\%diagnostic} magic
command.  \flowbook{} performs no analysis on those cells and offers
no rerun consistency guarantees for diagnostic cells.
\steve{A user should therefore mark a cell as diagnostic only when
it is used for visual inspection or debugging, and not as part of the main workflow.}

\subsection{Suggesting Fixes with an AI Assistant}
\label{sec:fix-assistant}

\note{\flowbook{} includes an AI
assistant that uses a user-configurable language model to propose and apply fixes for rerun consistency violations, as illustrated below.}

\smallskip
{\centering
\begin{nbnotebook}[0.75\columnwidth]{}
\begin{nbcodeerror}{X}{@S}
\begin{lstlisting}
x = x + 1
\end{lstlisting}
\end{nbcodeerror}
\begin{nberror}
Reads and writes \texttt{x}
\end{nberror}
\begin{nbsuggest}
Cell @S both reads and writes \texttt{x}, making it non-idempotent under rerun.
\end{nbsuggest}
\begin{nbfixrow}
\nbfixbutton{Rename written x to x\_plus\_one}\hspace{6pt plus 4pt}%
\nbfixbutton{Describe your own fix}
\end{nbfixrow}
\end{nbnotebook}
\smallskip
}

The assistant
explains why the cell is not rerun consistent and then offers one
or more candidate fixes that the scientist can apply with a single click.  Those fixes are
drawn from the taxonomy in Table~\ref{tab:repro-errors} below or generated by the assistant's own reasoning.
Alternatively, the scientist can describe a fix in their own words.  The assistant then applies the fix directly to the notebook and gives the scientist a chance to review the changes and either accept or cancel them, iterating as needed until they are satisfied.

%% file: nbs/mutation-example.tex
\def\mutw{\linewidth}%
\def\mutbegin{\begin{minipage}[t]{0.48\textwidth}\vspace{0pt}}%
\def\mutend{\end{minipage}}%
\def\mutsep{\hfill}%
\mutbegin
\begin{nbnotebook}[\mutw]{(a) Element-wise update causes inconsistent output}
\nbflow{nbmutA}
\begin{nbcode}{1}{@A}
\begin{lstlisting}
df = pd.DataFrame({"x": [0,0,0]})
\end{lstlisting}
\end{nbcode}

\begin{nbcode}{2}{@B}
\begin{lstlisting}
df.loc[0, "x"] = 4
\end{lstlisting}
\end{nbcode}

\begin{nbedit}{3}{4}{@C}
\begin{lstlisting}
|df.loc[1, "x"] = 5
\end{lstlisting}
\nbeditsep
\begin{lstlisting}
df.loc[2, "x"] = 6
\end{lstlisting}
\end{nbedit}

\begin{nbcode}{5}{@D}
\begin{lstlisting}
df["x"].tolist()
\end{lstlisting}
\end{nbcode}
\begin{nboutput}{5}
[4, 5, 6] \ \textcolor{red}{\# top-to-bottom: [4, 0, 6]}
\end{nboutput}
\end{nbnotebook}
\nbflowarrow[arc=2mm]{nbmutA}{1}{2}
\nbflowarrow[arc=2mm]{nbmutA}{2}{3}
\nbflowarrow[arc=2mm]{nbmutA}{3}{4}
\nbflowarrow[arc=2mm]{nbmutA}{4}{5}

\mutend
\mutsep
\mutbegin
\begin{nbnotebook}[\mutw]{(b) \flowbook{} requires full-column assignment}
\nbflow{nbmutB}
\begin{nbcode}{1}{@A}
\begin{lstlisting}
df = pd.DataFrame({"x": [0,0,0]})
\end{lstlisting}
\end{nbcode}

\begin{nbcodeerror}{X}{@B}
\begin{lstlisting}
df.loc[0, "x"] = 4
\end{lstlisting}
\end{nbcodeerror}
\begin{nberror}
Use \texttt{df["x"] = ...} for full-column assignment
\end{nberror}

\vspace*{0.5em}

\begin{nboutput}{Y}
...
\end{nboutput}

\vspace*{0.5em}

\end{nbnotebook}
\nbflowarrow[arc=2mm]{nbmutB}{1}{X}
\mutend

%% file: evaluation.tex
\section{Evaluation}
\label{sec:evaluation}

Our evaluation addresses three research questions:
\begin{itemize}
  \item \textbf{RQ1:} \steve{How prevalent are rerun consistency violations in
  real-world notebooks, and how often do they lead to output
  consistency violations?}
\item \textbf{RQ2:} Can the rerun consistency violations reported by
  \flowbook{} be readily addressed?
\item \textbf{RQ3:} Is the \flowbook{} analysis efficient enough for
  interactive use?
\end{itemize}

To answer these questions, we constructed a benchmark suite from the
Kaggle Playground Series competitions for January--June 2025.  These
are monthly machine learning competitions open to all
users~\cite{kaggle_playground_series}. Each competition involves a
supervised learning task on tabular data. Submitted
notebooks implement complete ML pipelines including data loading,
exploratory analysis, feature engineering, model training, and
prediction.

We selected the top-voted public notebooks from each competition that
met the following inclusion criteria: (1)~the notebook uses only
publicly available data and code; (2)~the notebook executes without
error in a standardized virtual environment with fixed Python library
versions supported by \flowbook for many of the most commonly used
data science libraries\footnote{We did, however, fix several pandas
operations that would behave incorrectly when executed in the presence
of copy-on-write semantics, which were disabled by default in the
Kaggle environment.  We also merged two cells in one notebook that shared
a matplotlib \t{Axes} object because \flowbook{}'s current implementation is unable to run a notebook with shared
non-deep-copyable objects (Section~\ref{sec:checkpointing}) to completion.}; and (3)~the baseline execution time is under one
hour.  These criteria ensure reproducible experimental conditions
while excluding notebooks that depend on private datasets or pose
challenges to our experimental setup.

\data{The resulting benchmark comprised 61 notebooks spanning a
representative range of programming expertise, ML techniques, and
library usage.} Notebooks employed production libraries including pandas~\cite{mckinney2010},
NumPy~\cite{harris2020numpy},
RAPIDS~\cite{rapids2026},
scikit-learn~\cite{DBLP:journals/jmlr/PedregosaVGMTGBPWDVPCBPD11},
XGBoost~\cite{DBLP:conf/kdd/ChenG16},
CatBoost~\cite{DBLP:conf/nips/ProkhorenkovaGV18},
LightGBM~\cite{DBLP:conf/nips/KeMFWCMYL17},
Keras~\cite{2018ascl.soft06022C}, and
PyTorch~\cite{DBLP:conf/nips/PaszkeGMLBCKLGA19}.
\data{Our experimental infrastructure ran notebooks via a headless
\flowbook{} driver on a compute server allocating four CPUs, one GPU
(NVIDIA RTX A6000), and 16 GB memory to each notebook.}

\data{Notebooks contained between 3 and 51 code cells with executable Python statements (median: 13, total: 999); cell run
times varied from under 1 ms to 52 minutes (median: 117 ms); total execution time varied
from under 2 seconds to 52 minutes (median: 98 seconds); and peak memory usage varied from
under 1 MB to about 6 GB (median: 108 MB), plus up to 1 GB of GPU memory in 17 notebooks
(and one reserved the entire 48 GB GPU).}

\begingroup
\footnotesize
  \def\reprocatwidth{2.6cm}%
  \def\reprodescwidth{\dimexpr\textwidth-2.6cm-4\tabcolsep\relax}%
\begin{table}[t]
  \caption{Rerun consistency violation categories with examples and
    recommended localized fixes.}
  \label{tab:repro-errors}
  \centering
  \begin{tabular}{>{\raggedright\arraybackslash\bfseries}m{\reprocatwidth}>{\raggedright\arraybackslash}m{\reprodescwidth}}
  \toprule
  \textbf{Category} & \textbf{Description and Fix} \\
  \midrule
  \rowcolor{gray!12}
  In-place variable reassignment &
    \vspace{4pt}A cell reads and overwrites the same variable, making reruns accumulate changes.
    \par\vspace{2pt}\textit{Example:} \texttt{train = pd.concat([train, extra])}.
    \par\vspace{2pt}\textit{Fix:} Write to a new variable at each step (e.g., \texttt{train\_combined = ...})\vspace{4pt} \\[4pt]
  Invalid mutation &
    \vspace{4pt}A cell mutates an object allocated by an earlier cell (Section~\ref{sec:mutation}).
    \par\vspace{2pt}\textit{Example:} \texttt{df.drop(columns=[\textquotesingle{}id\textquotesingle{}], inplace=True)}, or \texttt{model.fit(X, y)} in a cell after \texttt{model}'s allocation.
    \par\vspace{2pt}\textit{Fix:} Use non-mutating alternatives (e.g., \texttt{df = df.drop(...)}), apply mutations to fresh copies of the original variables, or merge cells related to initializing an object's internal mutable state.
    \vspace{4pt} \\[4pt]
  \rowcolor{gray!12}
  Reusing variable for different purposes &
    \vspace{4pt}A variable is reassigned to hold semantically different data.
    \par\vspace{2pt}\textit{Example:} \texttt{model} used for scoring, then reassigned to a new model.
    \par\vspace{2pt}\textit{Fix:} Use distinct variable names for different purposes. \\
  Diagnostic inspection before mutation &
    \vspace{4pt}A read-only inspection cell captures a pre-transformation state.
    \par\vspace{2pt}\textit{Example:} A cell calls \texttt{df.info()}; a cell below writes to \texttt{df}.
    \par\vspace{2pt}\textit{Fix:} Move inspection calls after the transformation, or use \texttt{\%diagnostic} to mark the cell.
    \vspace{4pt} \\[4pt]
  \bottomrule
  \end{tabular}
\end{table}
\endgroup

\begin{figure}[tp!]
  \centering
  \includegraphics[width=0.65\textwidth]{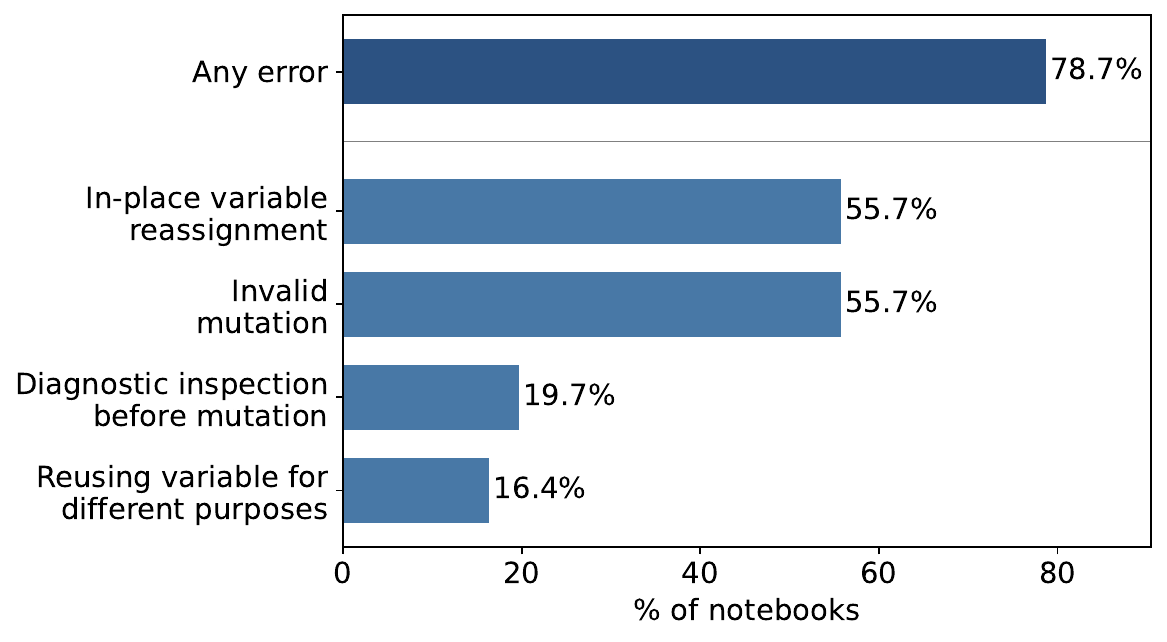}
  \caption{\label{fig:error-categories} \data{Rerun consistency
      violations are common in our Kaggle benchmark suite and fall
      into four categories} (\autoref{tab:repro-errors}).  Bars show
    the percentage of notebooks containing at least one instance of
    each category.  Many notebooks exhibit multiple violation types.}
\end{figure}

\subsection*{RQ1: Rerun Consistency Violations in Real-World Notebooks}
\label{sec:eval-violations}

We executed each notebook top to bottom under \flowbook{}, configured
to record rerun consistency violations without blocking cells.
\data{\flowbook{} flagged 203 cells in 48 of the 61 notebooks as rerun
  inconsistent.}  Because these cells were flagged during a
top-to-bottom execution, each violation is at this point a latent
hazard: the notebook remains output consistent until a flagged cell is
rerun.

\steve{ We reran each flagged cell in two ways after a top-to-bottom
  execution: rerunning the cell and every cell below it, and, for a
  cell that overwrites a value read by a cell above it, rerunning that
  earlier cell. For 126 of the 203 rerun-inconsistent cells (62\%) one
  of these reruns changes a visible output (text, HTML, error, figure,
  or written file), leaving the notebook output inconsistent: 83 cells
  change their own output, 29 change only a later cell's output or a
  file, and 14 change the output of the earlier cell that read the
  value they overwrote. (To exclude differences due to nondeterminism,
  we discount any difference that also appears when the flagged cell
  is not rerun, or when the earlier cell is rerun before the flagged
  cell has run.)
  These counts are lower bounds, because we did not explore all
  possible rerun orders of notebook cells. Nevertheless, these results
  show that most (at least 62\%) of rerun consistency violations are
  also violations of output consistency.  }

  \steve{A further 43 of the flagged cells (21\%) change notebook state without
  changing any visible output; such changes may still affect the
  results of later edits and reruns. Another 23 (11\%) are rerun
  inconsistent but cannot produce output inconsistency, because
  rerunning them rewrites the same values (for example, \texttt{df =
    df.dropna()}). The remaining 11 (5\%) are false positives of
  \flowbook{}'s analysis, caused by its conservative read-set tracking and
  treatment of mutations.}

Inspection of those 203 cells identified four basic categories of
rerun consistency violations (Table~\ref{tab:repro-errors}).  They
cover every violation in our study but may not be exhaustive for other
programs.  Figure~\ref{fig:error-categories} shows the fraction of
notebooks containing at least one violation of each category.

\smallskip
\begin{tcolorbox}[colback=green!5!white,colframe=green!75!black, arc=1mm, top=1mm, bottom=1mm, left=1mm, right=1mm]
\textbf{RQ1 Takeaway:} \data{Rerun consistency violations are common:
  79\% of notebooks (48 of 61) contain at least one, with 203 cells
  flagged in total.  These rerun consistency violations cluster into
  four categories: in-place reassignment and invalid mutation are the
  most common.  Most rerun consistency violations (at least 62\%) also
  violate output consistency.}
\end{tcolorbox}

\begin{figure*}[tp!]
  \centering
  \includegraphics[width=\textwidth]{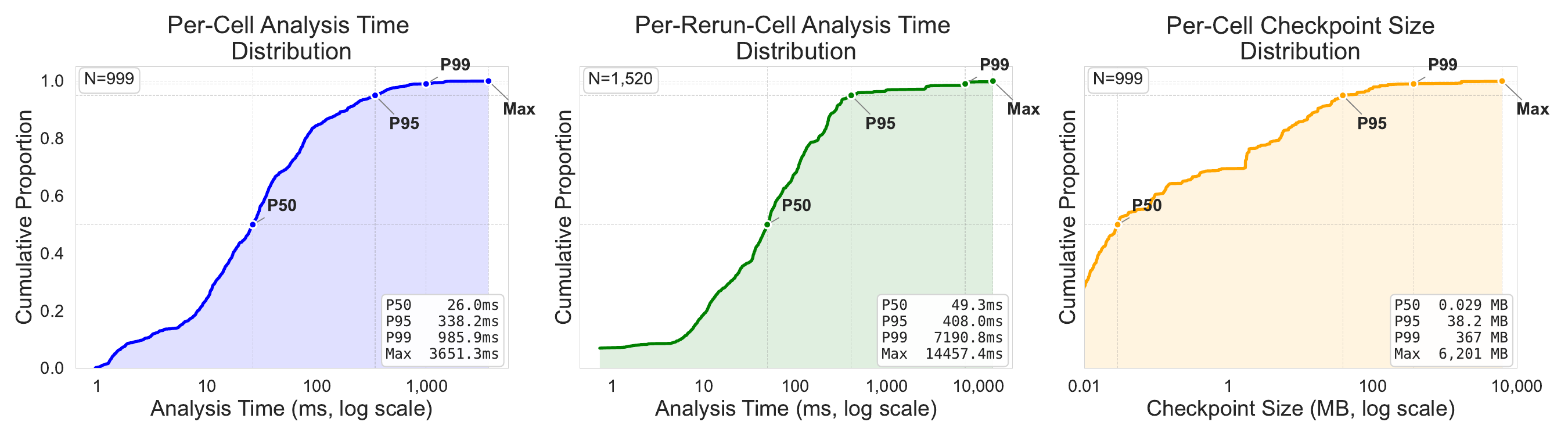}
  \caption{\label{fig:cdfs}Cumulative distributions of \flowbook{}'s per-cell overhead. We make the following judgments based on the median (P50) data points. \data{(Left) Per-cell analysis time in a top-to-bottom execution is barely perceptible. (Middle) Rerunning cells multiple times does not significantly increase analysis time. (Right) The near-zero per-cell checkpoint size reflects \flowbook{}'s optimized copy-on-write checkpointing.}}

\end{figure*}

\subsection*{RQ2: Automatically Repairing Rerun Consistency Violations}
\label{sec:eval-repair}

\steve{To determine whether the observed rerun consistency violations
  reflect fundamental incompatibilities with \flowbook{}'s programming
  discipline, we developed an automated program repair
  tool~\cite{DBLP:journals/cacm/GouesPR19} that equipped a Claude
  Code~\cite{claudecode} agent with a specification of rerun
  consistency, example violations and fixes for each category in Table~\ref{tab:repro-errors},
  and notebook-editing primitives. The agent fixed one violation at a
  time, rerunning the affected cells and trying another fix whenever a
  fix introduced new violations or raised an exception, until \flowbook{} reported
  every cell clean. We then checked each repair: an independent
  reviewer agent judged from the code whether it preserved behavior,
  and a rerun under \flowbook{} confirmed that no violations remained
  and that the fixed notebook produced the same outputs as the
  original. A failed check sent the findings back to the agent for
  another round.  \data{41 of the 48 rerun-inconsistent notebooks
    passed in the first round; the other 7 needed one more round,
    after the reviewer rejected one repair and the output comparison
    rejected six.  Only 9 of 203 repaired cells (4\%) needed more than
    one attempt. Each attempt took the agent a median of 14\,s, and at
    most 3.3 minutes, to diagnose a violation and edit the code, plus
    the time to rerun the affected cells.}  We also manually inspected
  every repair diff.}

\steve{Importantly, none of the repairs required changes to the notebook's
computational logic or produced measurable performance regressions. 
Some fixes were perhaps suboptimal, but all preserved behavior and runtime.
These results suggest that the
observed violations are fixable anti-patterns rather than fundamental
incompatibilities with \flowbook{}'s programming discipline. Further improving
automated rerun consistency repair remains an interesting direction for future
work.}

\smallskip
\begin{tcolorbox}[colback=green!5!white,colframe=green!75!black, arc=1mm, top=1mm, bottom=1mm, left=1mm, right=1mm]
\textbf{RQ2 Takeaway:} \data{An automated repair tool grounded in
  \flowbook{} fixed all 203 rerun consistency violations in 48
  notebooks, with each fix taking seconds to a few minutes.  All
  could be addressed with local fixes that do not change
  the logic or performance of the original notebooks.}  These results
indicate that rerun consistency is a viable programming discipline for output-consistent notebooks.
\end{tcolorbox}

\subsection*{RQ3: Performance Overhead}
\label{sec:eval-performance}

\myparagraph{Per-Cell Analysis Time Overhead} To assess the time
overhead of \flowbook's analysis, we measured the per-cell overhead
for all cells in all notebooks in a top-to-bottom execution.
Figure~\ref{fig:cdfs} (left) shows the cumulative distribution of
per-cell overhead. \data{The median (P50) overhead was 26\,ms, with P95 at
338\,ms and P99 at 986\,ms.}
\data{The 26\,ms median overhead is smaller than the 100--200\,ms baseline latency of executing even a trivial Jupyter cell, which is dominated by communication and rendering costs.} Prior work finds that delays below 100\,ms appear instantaneous~\cite{nielsen1994usability} and delays below 500\,ms preserve analysis quality~\cite{DBLP:journals/tvcg/LiuH14}. \data{Thus, \flowbook{}'s overhead is unlikely to affect user behavior for the vast majority of cells: 85\% of cells incur under 100\,ms of overhead and 98\% incur under 500\,ms.}

\data{The analysis overhead is primarily due to taking checkpoints at cell boundaries.}
\steve{The overhead therefore grows with the amount of state \flowbook{}
checkpoints, not with the cell's computation time. Thus, even long-running cells
see negligible relative overhead. The 34 cells that run longer than one minute
incurred a median overhead of 44\,ms and a maximum of 581\,ms, never more than
0.41\% of the cell's execution time.}

\myparagraph{Per-Cell Analysis Time Overhead for Reruns}
The above only considers the time overhead of the initial analysis of
a cell during a top-to-bottom execution.  To measure the time overhead
of rerunning a cell, we conducted the following experiment: we first executed a notebook
top to bottom and then repeatedly cycled
through and reran the cells at the top, 25\% of the way down, 50\% of
the way down, 75\% of the way down, and at the bottom of the notebook.
\data{Figure~\ref{fig:cdfs} (middle) shows the cumulative distribution of time
per cell (P50 $=$ 49\,ms, P95 $=$ 408\,ms, P99 $=$ 7{,}191\,ms).
A small number of cells incurred significant overhead.  Of the 48 reruns
exceeding 1\,s, 25 occur in a single notebook, due to the cost of
checkpointing and diffing a single 4.3\,GB matrix.}  Augmenting
\flowbook{} with a rudimentary static dependency analysis would likely
enable \flowbook{} to avoid copying that matrix on each rerun.

\myparagraph{Per-Cell and Peak Memory Overhead} Figure~\ref{fig:cdfs}
(right) shows the distribution of per-cell checkpoint size, the
additional memory a cell's checkpoint requires beyond its base state.  
\note{\data{The median per-cell checkpoint size was just 0.029\,MB (P95 $=$ 38.2\,MB,
P99 $=$ 367\,MB). Most cells incurred negligible memory cost because \flowbook{}
exploits copy-on-write semantics for DataFrames and immutability properties to
minimize the amount of memory copied.} \data{Our measurements include both CPU and GPU
memory.} \data{For notebooks using RAPIDS GPU DataFrames, \flowbook{}
checkpoints these DataFrames in place using the same copy-on-write optimizations as on the CPU.}}

\note{\data{Peak memory overhead per notebook, measured as the maximum heap
footprint under \flowbook{} relative to the baseline maximum heap footprint, was similarly
modest (median $=$ 4.8\%, P95 $=$ 57\%).}}

\vspace{-0.5em}
\smallskip
\begin{tcolorbox}[colback=green!5!white,colframe=green!75!black, arc=1mm, top=1mm, bottom=1mm, left=1mm, right=1mm]
\textbf{RQ3 Takeaway:} \data{\flowbook{} imposes imperceptible or nearly imperceptible time overhead (median $=$ 26\,ms) for
interactive use.}
\data{Per-cell checkpoint size (median $=$ 0.029\,MB) is also negligible in all but the most extreme
memory-intensive cases.}
\end{tcolorbox}

%% file: related.tex
\section{Related Work}
\label{sec:related}

\myparagraph{Computational Notebooks and Reproducibility}
Computational notebooks, most notably Jupyter~\cite{jupyter}, are
widely used for data analysis because they support rapid iteration and
exploration~\cite{DBLP:conf/chi/KeryRAJM18}. But that same flexibility
undermines reproducibility and
reliability~\cite{DBLP:conf/chi/KeryRAJM18,DBLP:conf/chi/RuleTH18,pimentel2021understanding,10.1093/gigascience/giad113,DBLP:conf/icse/HuangRHTLC25,quaranta2022collabNotebooks}.
In a study of 1.4 million GitHub notebooks, Pimentel et al. found only
4\% produced identical results on
re-execution~\cite{pimentel2021understanding}, and Wang et al. found
73\% of published notebooks fail to
reproduce~\cite{DBLP:conf/kbse/WangKLZ20}. One source of these failures
is output inconsistency, where a notebook's saved outputs no longer
match what a top-to-bottom execution of its code would produce. The execution model
invites this by permitting mutable shared state and non-linear execution.

Many user-interface and programming-model enhancements improve
notebook organization and comprehension, without directly enforcing output
consistency.  These include grouping and hiding cells through
``cell folding''~\cite{DBLP:journals/pacmhci/RuleDTH18}; grouping and annotating
cells to support sensemaking~\cite{DBLP:conf/chi/ChattopadhyayPH20};
laying cells out two-dimensionally~\cite{harden2023twoDNotebooks};
end-user-defined scoping of variables~\cite{DBLP:conf/chi/RawnC25}; and visualizing cell dependencies
in a notebook~\cite{wenskovitch2019albireo}. 

\myparagraph{Dependency-based Approaches}
Several systems detect stale cells using static or dynamic dependency
analysis.  NodeBook~\cite{nodebook} tracks cell dependencies to flag
and rerun stale cells. Dataflow
notebooks~\cite{koop2017dataflowNotebooks} generate dataflow graphs
from notebooks to help users understand cell dependencies.

\note{Osiris~\cite{DBLP:conf/kbse/WangKLZ20} takes a post-hoc approach to
reconstruct the execution that led to a notebook's recorded outputs. This is a
complementary approach to \flowbook{}, which ensures
that the current notebook order remains consistent with a top-to-bottom
execution by enforcing rerun consistency throughout interactive development.}

\note{The Code Gathering Tools~\cite{DBLP:conf/chi/HeadHBDD19} use static analysis to
identify code dependencies for notebook cleanup and post-hoc reorganization. In
contrast, \flowbook{} operates during interactive editing and execution,
preventing rerun consistency violations from ever entering notebooks, rather
than attempting to repair them after the fact. }
\note{IPyflow~\cite{DBLP:journals/pvldb/ShankarMCHP22} tracks dependencies
dynamically to detect staleness. However, as Section~\ref{sec:introduction} shows, it guarantees neither rerun nor output consistency.}

\note{Static analyses, including those mentioned above, cannot track all
dependencies through arbitrary library code and dynamic language features.  This
limitation may, for example, lead to failed execution order reconstruction in Osiris or
missed dependencies in the other systems. IPyflow and other dynamic/hybrid
analyses can avoid some of these problems, but only at a potentially prohibitive performance cost.  For example,
the current version of IPyflow restricts dynamic analysis of loops to only the
first iteration to mitigate high overheads.  Such restrictions may 
miss dependencies.} %

\myparagraph{Reactive Dataflow Models for Notebooks}
Reactive notebook systems enforce dataflow execution models that
automatically re-execute cells when dependencies change, with
implementations for JavaScript~\cite{observablehq2026},
Julia~\cite{plutojl2026}, and
Python~\cite{marimo2026,guzharina2021reactive,hextech2026}, among
others. While these provide stronger ordering guarantees than standard
notebooks, none guarantee output consistency for Python's scientific
ecosystem~\cite{arXiv:2511.21994}: systems that infer dependencies may
miss true ones, so recorded outputs can diverge from a top-to-bottom
execution.
Section~\ref{sec:introduction} demonstrates this for Marimo~\cite{marimo2026}.
Unlike reactive systems, whose guarantees hold only insofar as their
dependency inference is complete, \flowbook's guarantee is independent
of dependency inference for mutable state.

\myparagraph{Notebook State Management}
Prior work on notebook checkpointing has focused on state versioning
and execution backtracking. ForkIt~\cite{DBLP:conf/chi/WeinmanDBD21}
supports branching and rollback by serializing the full Python
interpreter state with dill~\cite{dill}, but full-state serialization
incurs substantial overhead. Kishu~\cite{DBLP:journals/pvldb/LiCSPS24}
reduces storage and restoration costs by tracking incremental
kernel-state changes between cell executions, but still relies on
expensive serialization to materialize and compare checkpoints. Both
target state versioning rather than output consistency. In
contrast, \flowbook{} employs lightweight in-memory checkpointing. 
Maintaining checkpoints as deep copies of the global
environment avoids serialization entirely and exploits characteristics
of common notebook workloads to achieve low overhead
(Section~\ref{sec:evaluation}).

%% file: threats.tex
\section{Threats to Validity}
\label{sec:threats}

\textbf{Construct validity:} Rerun consistency is deliberately stronger than
output consistency, so some reported violations may be benign: in our rerun
experiment, 38\% of flagged cells changed no visible output, typically because
they modified only notebook state or rewrote identical values; only 5\% were
false positives arising from \flowbook{}'s conservative read-set and mutation
tracking. The quality of AI-generated repairs may also depend on the prompt,
model, and available context.
\textbf{Internal and ecological validity:} Because \flowbook{} is a dynamic
analysis, our results reflect one top-to-bottom execution per notebook in a
fixed environment. Unexecuted paths, different inputs, interactive histories,
or library versions could expose additional violations or change overheads.
Our latency measurements and experience suggest that \flowbook{} preserves an
interactive workflow, but a formal usability study remains future work.

\textbf{External validity:} Our benchmark suite covers tabular, supervised
machine-learning notebooks from a single source and may not represent all
notebook uses. Each notebook must run end to end in a fixed environment with
public data, so that baseline and \flowbook{} runs can be compared, and must
exercise the libraries \flowbook{} supports. Kaggle notebooks meet these
requirements because they ship with their data and target a standard
environment.

%% file: conclusion.tex
\section{Conclusion}
\label{sec:conclusion}

We presented \flowbook, a Jupyter plugin whose dynamic analysis enforces
rerun consistency in computational notebooks. \flowbook{} maintains a
well-formedness invariant guaranteeing that when all cells are clean,
the notebook's visible outputs match those of a top-to-bottom execution. 
\flowbook{} enforces
rerun consistency in real-world notebooks without compromising
interactivity. \data{Across 61 Kaggle notebooks, \flowbook{} flagged 203
rerun-inconsistent cells, at least 62\% of which violate output
consistency when rerun. \flowbook{}'s AI repair agent eliminated every violation with
local edits, and \flowbook{} adds a median overhead of only 26\,ms per
cell execution.}

%% file: acknowledgements.tex
\newcommand{\aidisclosure}{

\textbf{Implementation:} We used Claude Code (various Opus and Fable
models) to generate portions of \flowbook{}'s implementation and
documentation.  We reviewed all generated code and validated it with
test suites, collections of notebooks, and benchmarks.
\textbf{Evaluation:}
Claude Code generated parts of the scripts that execute the benchmarks,
process \flowbook{}'s output, compute the reported statistics, and produce
the evaluation figures.
We reviewed those scripts, spot-checked reported values against raw execution
logs, and independently verified key results.
\textbf{Violation coding:} We developed the categories in
Table~\ref{tab:repro-errors} iteratively over several rounds.  In each
round, a Claude Code (Opus 5.5) agent classified all 203 rerun-consistency
violations, and we reviewed its classifications and refined the
category definitions.  We reviewed every final classification.
\textbf{Automated repair:} Our
automated program repair tool equips a Claude Code (Opus 5.5) agent with a
specification of rerun consistency, example violations and fixes, and
notebook-editing primitives.  A reviewer agent judged whether each
repair preserved behavior.  It ran in a separate session with no
access to the repair agent's reasoning.  A subsequent rerun under
\flowbook{} was used to confirm that no violations remained and that the
outputs matched the original.  We also
manually inspected every repair diff.
\textbf{\flowbook{} assistant:}
\flowbook{}'s repair assistant uses a user-configurable language model.
We produced the example in Section~\ref{sec:implementation} with Claude Opus 5.1.
\textbf{Lean proofs:}
Claude Code (Opus 5.1--5.5) produced the initial Lean~4 formalization from the initial text of
the supplement, which we wrote. 
We checked each definition and theorem statement against the paper, and made
several small changes to improve the clarity and precision of our supplement.
The Lean kernel checks every proof, so the proofs' validity does not depend
on the AI.

}

\newcommand{\datavailability}{The \flowbook{} system will be made publicly available upon publication.
It is a fully featured system with an IPython kernel and a Jupyter
notebook extension.  The source code, tutorials, documentation, and Lean proof of
correctness will be released on GitHub and
distributed via PyPI. \flowbook{} will also be submitted to the Artifact
Evaluation Committee (AEC) to support reproducibility and independent
validation.  A preliminary version can be accessed at: \url{https://anonymous.4open.science/r/FlowBook/}.
The notebook dataset used in this work is derived from publicly
available Kaggle notebooks. These notebooks have been curated and
packaged into a reusable repository to facilitate replication and
reuse in future research projects.  That repository will be made publicly available upon publication as well.
A preliminary version of the dataset and our evaluation infrastructure can be accessed at: \url{https://anonymous.4open.science/r/FlowBookEval}.}

\ifarxiv

\section{Acknowledgments}

This material is based upon work supported by the National Science Foundation under Grant Nos. 2243636 and 2243637. Any opinions, findings,
and conclusions or recommendations expressed in this material are
those of the author(s) and do not necessarily reflect the views of
the National Science Foundation.

\section{Availability}

\flowbook{} is available at: \url{https://github.com/stephenfreund/FlowBook}.

\else

\section{Use of AI Tools}
\aidisclosure

\section{Data Availability}
\steve{\datavailability}

\fi

%% file: supplement-body.tex
\newtheorem{assumption}[theorem]{Assumption}

\ifarxiv These appendices present\else This supplemental material presents\fi{} the full formal development
underlying \flowbook{}'s analysis. 
Section~\ref{sec:semantics} formalizes an idealized model of
computational notebooks, defines the notebook operations, and gives the
reference top-to-bottom execution that grounds the notion of an
output-consistent state (the formal counterpart of an output-consistent
notebook in \ifarxiv Section~\ref{sec:analysis}\else the main paper\fi).
Section~\ref{sec:analysis-supp} presents the instrumented semantics:
the read, write, and status bookkeeping; the rerun consistency
predicates; the forward and backward staleness predicates; the
well-formedness invariant; and the statements of the preservation,
output-consistency, and progress theorems with proof sketches.
Sections~\ref{sec:preservation-proof}--\ref{sec:progress-proof} give
the full proofs of those three theorems.  Section~\ref{sec:lean} describes
our mechanization of these proofs in Lean 4.

\section{A Semantics for Interactive Computational Notebook Execution}
\label{sec:semantics}

This section formalizes an idealized model of computational notebooks
and defines output consistency.

\subsection{Notebook Model}

A \emph{notebook} consists of a sequence of cells, each containing
source code $c \in \mathit{Code}$ and outputs $o \in
\mathit{Output}$. We leave the exact form of the cell source code and
output unspecified, as it is not relevant to the semantics.  In
addition to those components, which are visible to the user, the
notebook state also includes a hidden execution store used during cell
evaluation.  Thus, a notebook state is a tuple $S = (C, O, \Sigma)$,
where:

\begin{itemize}
  \item $C = C_1, \ldots, C_n$ is the sequence of cell source code,
  \item $O = O_1, \ldots, O_n$ is the sequence of cell outputs, and
  \item $\Sigma : \Loc \to \Val$ is the current store.
\end{itemize}
We write $C_i$ for the source code of cell $i$, $O_i$ for its most
recent output, and $\Sigma(\ell)$ for the value of location $\ell$ in
store $\Sigma$. An output $O_i = \bot$ indicates cell $i$ has not yet
been executed.

\noindent
\myparagraph{Locations} We model state at two location granularities: top-level variables
and individual columns for DataFrames. A \emph{location}
$\ell \in \Loc$ identifies a unit of state that cells may read or
write:
\[
\ell \in \Loc ::= x \mid d.c
\]
where $x \in \mathit{Var}$ denotes a top-level variable binding, $d
\in \mathit{Address}$ denotes a DataFrame address, and $c \in
\mathit{String}$ denotes column $c$ of the DataFrame at address $d$.

This location granularity reflects a key observation about notebook
workflows: Cells communicate through global variables, and DataFrame
columns are the natural unit of independent modification.  When
reasoning about inter-cell dependencies, this formulation enables our analysis
to identify when two cells read or write to disjoint columns of the
DataFrame at address $d$.  In contrast, our analysis treats all other
structures, such as lists, dictionaries, and arrays, as atomic values
stored at a single location, and all reads and updates operate on these
values.

We assume DataFrame addresses are immutable: address $d$ always refers
to the same DataFrame object.  An operation like \<df['price'] = 100>
reads the DataFrame address $d$ from the global variable \texttt{df},
and writes to the column \texttt{price} of the DataFrame at address
$d$.

\noindent
\myparagraph{Cell Evaluation} We define cell evaluation using a big-step judgment that captures the
effect of executing a single cell.  The judgment
\[
\StdEvalCell{c}{\Sigma}{o}{\Sigma'}
\]
asserts that
executing the code $c$ in store $\Sigma$ produces output $o$ and
resulting store $\Sigma'$.

We treat cell evaluation as a black box that models the underlying
language runtime (Python, in our implementation) while
abstracting over the details of expression evaluation, statement
execution, library calls, and calls into native code (e.g., C extensions).
Because a cell may invoke non-deterministic operations, this
judgment is a \emph{relation} rather than a function.

\subsection{Notebook Operations: Syntax and Semantics}

Users interact with notebooks through operations that run cells, edit
their contents, or modify notebook structure:
\[
op ::= \Run(i) \mid \Edit(i,\, c) \mid \Insert(i,\, c) \mid \Delete(i) \mid \Move(i,\, j)
\]
where:
\begin{itemize}[nosep]
\item $\Run(i)$ executes cell $i$, updating its output and the store.
\item $\Edit(i, c)$ replaces the source code of cell $i$ with $c$.
\item $\Insert(i, c)$ inserts a new cell with code $c$ at position $i$,
  shifting cells $i$ through $n$ to positions $i{+}1$ through $n{+}1$.
\item $\Delete(i)$ removes the cell at position $i$, shifting
  cells $i{+}1$ through $n$ to positions $i$ through $n{-}1$.
\item $\Move(i, j)$ relocates the cell at position $i$ to position $j$.
\end{itemize}

Figure~\ref{fig:std-semantics} specifies how each operation transforms
notebook state. In these rules and below, we write $X_{i..j}$ to denote the subsequence $X_i, \ldots, X_j$, and
$X[i := v]$ to denote the sequence $X$ with position $i$ updated to
$v$. We use $\emptyset$ for the empty store.

The key observation is that \textsc{[Std-Run]} permits
executing any cell at any time, regardless of whether cells above it
have executed. This permissiveness models interactive notebook use,
where users routinely execute cells out of order.
The remaining rules are straightforward. \textsc{[Std-Edit]} updates
the source code without affecting the store. \textsc{[Std-Insert]} and
\textsc{[Std-Delete]} add or remove cells, shifting indices
accordingly. \textsc{[Std-Move]} relocates a cell by composing delete
and insert.

\begin{figure}[tp!]
  \centering
  \newcommand{\StdRunRule}{\nrln{Std-Run}{
      \StdEvalCell{C_i}{\Sigma}{o}{\Sigma'} \qquad
      O' = O[i := o]
    }{
      (C,\, O,\, \Sigma)
      \;\xrightarrow{\Run(i)}\;
      (C,\, O',\, \Sigma')
    }}
  \newcommand{\StdEditRule}{\nrln{Std-Edit}{
      C' = C[i := c]
    }{
      (C,\, O,\, \Sigma)
      \;\xrightarrow{\Edit(i,\, c)}\;
      (C',\, O,\, \Sigma)
    }}
  \newcommand{\StdInsertRule}{\nrln{Std-Insert}{
      C' = C[i := c] \quad
      O' = O[i := \bot]
    }{
      (C,\, O,\, \Sigma)
      \;\xrightarrow{\Insert(i,\, c)}\;
      (C',\, O',\, \Sigma)
    }}
  \newcommand{\StdDeleteRule}{\nrln{Std-Delete}{
      C' = C_{1..i-1} \cat C_{i+1..n} \\
      O' = O_{1..i-1} \cat O_{i+1..n}
    }{
      (C,\, O,\, \Sigma)
      \;\xrightarrow{\Delete(i)}\;
      (C',\, O',\, \Sigma)
    }}
  \newcommand{\StdMoveDownRule}{\nrln{Std-Move-Down}{
      s < d \\
      S \;\xrightarrow{\Delete(s)}\; S'' \;\xrightarrow{\Insert(d-1,\, C_s)}\; S'
    }{
      S \;\xrightarrow{\Move(s,\, d)}\; S'
    }}
  \newcommand{\StdMoveUpRule}{\nrln{Std-Move-Up}{
      s > d \\
      S \;\xrightarrow{\Delete(s)}\; S'' \;\xrightarrow{\Insert(d,\, C_s)}\; S'
    }{
      S \;\xrightarrow{\Move(s,\, d)}\; S'
    }}
  \[
  \begin{array}{@{}l@{\qquad}l}
  \myheadrule{S \xrightarrow{op} S' \strut} \\[1ex]
  \begin{array}{@{}l@{}}
    \StdRunRule
    \rulegap
    \StdInsertRule
    \rulegap
    \StdMoveDownRule
  \end{array}
  &
  \begin{array}{@{}l@{}}
    \StdEditRule
    \rulegap
    \StdDeleteRule
    \rulegap
    \StdMoveUpRule
  \end{array}
  \end{array}
  \]
  \caption{
    The standard semantics operates on notebook state $S = (C, O, \Sigma)$ where
    $C$ is the sequence of cell source code, $O$ is the sequence of cell outputs, and $\Sigma$ is the variable store.}
  \label{fig:std-semantics}
\end{figure}

\subsection{Output Consistency}
\label{sec:output-consistency-supp}

The flexibility of the standard semantics is the source of
output-consistency problems. Because users may execute cells in any
order, the current notebook state may diverge from what a clean
top-to-bottom execution would produce.

We formalize the reference behavior using a top-to-bottom execution
judgment. The judgment
\[
\nrln{Std-Top-to-Bottom-Execution}{
  \StdEvalCell{C_1}{\emptyset}{O_1}{\Sigma_1}
  \qquad
  \StdEvalCell{C_2}{\Sigma_1}{O_2}{\Sigma_2}
  \qquad
  \ldots
  \qquad
  \StdEvalCell{C_n}{\Sigma_{n-1}}{O_n}{\Sigma_n}
}{
  \StdEvalNotebook{C}{O}{\Sigma_n}
}
\]
asserts that executing cells $C_1, \ldots, C_n$ top-to-bottom from the
empty store produces outputs $O_1, \ldots, O_n$ and final store
$\Sigma$ via the intermediate stores $\Sigma_1, \ldots,
\Sigma_{n-1}$.

\begin{definition}[Output-Consistent State]
  \label{def:output-consistent-supp}
  A notebook state $S = (C, O, \Sigma)$ is \emph{output consistent} if
  there exists $\Sigma'$ such that $\StdEvalNotebook{C}{O}{\Sigma'}$.
\end{definition}
That is, a notebook state is output consistent if and only if
some top-to-bottom execution of all cells from an empty store
produces exactly the outputs currently recorded in the notebook. The final stores $\Sigma$ and
$\Sigma'$ need not be identical. $\Sigma$ reflects the actual
interactive execution history while $\Sigma'$ reflects the reference
execution, but the outputs $O$ must agree.

\section{Dynamic Analysis for Output Consistency}
\label{sec:analysis-supp}

\begin{figure}[tp!]
  \centering
  \footnotesize
  \[
  \begin{array}{@{}l@{\quad}l@{}}
  \myheadrule{S \cdot I \xRightarrow{op} S' \cdot I'}\\
  \begin{array}{@{}l@{}}
    \nrln{Inst-Run}{
      S = (C, O, \Sigma) \qquad
      C_i; \Sigma \,\Downarrow\, o \cdot \Sigma' \cdot r \cdot w \\
      S' = (C, O[i := o], \Sigma') \\
      R' = R[i := r] \quad W' = W[i := w] \\
      R' \text{ and } W' \text{ are rerun consistent for } i \\
      T'_j =
      \left\{\begin{array}{@{}l@{\;}l@{}}
        \clean & \text{if } j = i \\
        \stale & \text{if } \ForwardStale(R, W, W', i, j) \\
        \stale & \text{if } \BackwardStale(W, W', i, j) \\
        T_j    & \text{otherwise}
      \end{array}\right.
    }{
      S \cdot (T,\, R,\, W)
      \;\xRightarrow{\Run(i)}\;
      S' \cdot (T',\, R',\, W')
    }
    \rulegap
    \nrln{Inst-Insert}{
      S \;\xrightarrow{\Insert(i,\, c)}\; S' \\
      R' = R[i := \emptyset] \quad
      W' = W[i := \emptyset] \quad
      T' = T[i := \stale]
    }{
      S \cdot (T,\, R,\, W)
      \;\xRightarrow{\Insert(i,\, c)}\;
      S' \cdot (T',\, R',\, W')
    }
    \rulegap
    \nrln{Inst-Move-Down}{
      s < d \\
      S \cdot I \;\xRightarrow{\Delete(s)}\; S'' \cdot I'' \;\xRightarrow{\Insert(d{-}1,\, C_s)}\; S' \cdot I'
    }{
      S \cdot I \;\xRightarrow{\Move(s,\, d)}\; S' \cdot I'
    }
    \end{array}
    &
    \begin{array}{@{}l@{}}
      \nrln{Inst-Edit}{
        S \;\xrightarrow{\Edit(i,\, c)}\; S' \qquad
        T' = T[i := \stale]
      }{
        S \cdot (T,\, R,\, W)
        \;\xRightarrow{\Edit(i,\, c)}\;
        S' \cdot (T',\, R,\, W)
      }
    \rulegap
    \nrln{Inst-Delete}{
      S \;\xrightarrow{\Delete(i)}\; S' \\
      R'' = R[i := \emptyset] \qquad
      W'' = W[i := \emptyset] \\
      T''_j =
      \left\{\begin{array}{@{}l@{\;}l@{}}
        \stale & \text{if } \ForwardStale(R, W, W'', i, j) \\
        \stale & \text{if } \BackwardStale(W, W'', i, j) \\
        T_j    & \text{otherwise}
      \end{array}\right. \\
      R' = R''_{1..i-1} \cat R''_{i+1..n} \qquad
      W' = W''_{1..i-1} \cat W''_{i+1..n} \\
      T' = T''_{1..i-1} \cat T''_{i+1..n}
    }{
      S \cdot (T,\, R,\, W)
      \;\xRightarrow{\Delete(i)}\;
      S' \cdot (T',\, R',\, W')
    }
    \rulegap
    \nrln{Inst-Move-Up}{
      s > d \\
      S \cdot I \;\xRightarrow{\Delete(s)}\; S'' \cdot I'' \;\xRightarrow{\Insert(d,\, C_s)}\; S' \cdot I'
    }{
      S \cdot I \;\xRightarrow{\Move(s,\, d)}\; S' \cdot I'
    }
  \end{array}
  \end{array}
  \]
  \caption{
    The instrumented semantics operates on $S \cdot I$.
    $I = (T, R, W)$ tracks cell status, read sets, and write sets.}
  \label{fig:inst-semantics}
\end{figure}

This section presents \flowbook's dynamic analysis for output consistency as
an instrumented semantics.  It
ensures output consistency by tracking read and write sets of locations
for each executed cell, enforcing rerun consistency, and maintaining
cell status (either $\stale$ or $\clean$).  We prove that this
instrumentation preserves a well-formedness invariant and that
well-formed states with all clean cells are output consistent.

\subsection{Instrumentation State}

Figure~\ref{fig:inst-semantics} presents the instrumented semantics
for an instrumented state $S \cdot I$ that augments the standard state
$S$ with \emph{instrumentation} $I = (T, R, W)$ where:
\begin{itemize}[nosep]
\item $T_i \in \{\clean, \stale\}$ is a \emph{tag} that records the status of cell $i$,
\item $R_i \subseteq \Loc$ is the set of locations read
  by cell $i$, and
\item $W_i \subseteq \Loc$ is the set of locations written by cell $i$.
\end{itemize}
The status $T_i = \clean$ indicates that cell $i$'s recorded output
remains valid with respect to the current store. The status
$T_i = \stale$ indicates that the cell must be re-executed. Initially,
all cells are stale with empty read and write sets.

\subsection{Instrumented Cell Evaluation}

Our instrumented semantics is based on an instrumented cell evaluation judgment
\[
\InstEvalCell{c}{\Sigma}{o}{\Sigma'}{r}{w}
\]
that extends standard cell evaluation
$
\StdEvalCell{c}{\Sigma}{o}{\Sigma'}
$
to also record the set of read locations $r \subseteq \Loc$ in $\Sigma$
that are read during evaluation of $c$, and the set of write locations
$w \subseteq \Loc$ where $\Sigma'$ is updated from $\Sigma$.
Note that a cell such as
\[
\texttt{x = 0;  x = x + 1}
\]
does not technically read the initial value of \texttt{x} in the
store, so this cell would have an empty read set $r = \emptyset$ and a
write set $w = \{ \texttt{x} \}$.  The \flowbook{} implementation
instruments the underlying run-time system to capture these sets.

The \textsc{[Inst-Run]} rule uses this extended judgment to evaluate
cell $i$ via $\InstEvalCell{C_i}{\Sigma}{o}{\Sigma'}{r}{w}$. That rule
requires that the read and write sets satisfy rerun consistency for the cell, meaning
that they capture accesses are rerun consistent, which is formally defined as follows.
\begin{definition}[Rerun Consistent Accesses] \label{def:rerun-consistent-accesses-supp}
  Instrumentation state $R$ and $W$ are \emph{rerun consistent} for cell $i$
  if these predicates hold:
  \begin{itemize}
    \item $\NoReadAndWrite(R, W, i)$: cell $i$ does not read and write the same location.
    \item $\WriteBeforeRead(R, W, i)$: every location read by $i$ was written by a cell above $i$.
    \item $\NoReadBeforeWrite(R, W, i)$: cell $i$ does not read locations written by cells below $i$.
    \item $\NoWriteAfterRead(R, W, i)$: cell $i$ does not overwrite locations read by cells above $i$.
  \end{itemize}
\end{definition}
We define these predicates formally as follows; \ifarxiv
Figure~\ref{fig:litmus-tests} illustrates them.\else the main paper
illustrates them in Figure~\ref{fig:litmus-tests}.\fi
\begin{eqnarray*}
  \NoReadAndWrite(R, W, i) &\defeq & R_i \cap W_i = \emptyset \\
  \WriteBeforeRead(R, W, i) &\defeq & R_i \subseteq (\cup\, W_{1..i-1}) \\
  \NoReadBeforeWrite(R, W, i) &\defeq & R_i \cap (\cup\, W_{i+1..n}) = \emptyset \\
  \NoWriteAfterRead(R, W, i) &\defeq &  W_i \cap (\cup\, R_{1..i-1}) = \emptyset
\end{eqnarray*}

The \textsc{[Inst-Run]} rule also sets to stale any cells whose inputs may have
been invalidated by the executed cell. That includes cells $C_j$ below $i$ whose
reads $R_j$ or writes $W_j$ conflict with the executed cell's writes $W'_i$ or
with locations $W_i \setminus W'_i$ that used to be written by cell $i$ but no longer are:
\begin{align*}
  \ForwardStale(R, W, W', i, j) &\defeq j > i \land (W_i \cup W'_i) \cap (R_j \cup W_j) \neq \emptyset
\end{align*}
If a cell $i$ previously wrote to a location $\ell$ but no longer does
(i.e., $\ell \in W_i \setminus W'_i$) then we also mark the nearest
writer of $\ell$ above cell $i$ as stale. This ensures that the write is
re-executed to produce a store consistent with
top-to-bottom execution.
\begin{align*}
  \LastWriter(W, i, \ell) &\defeq \max \{ k < i \mid \ell \in W_k \} \\
  \BackwardStale(W, W', i, j) &\defeq j < i \land j = \LastWriter(W, i, \ell) \text{ for some } \ell \in W_i \setminus W'_i
\end{align*}

The \textsc{[Inst-Delete]} rule uses the same logic to identify stale cells after
deleting cell $i$: these are cells whose reads or writes conflict with
what the deleted cell had written.
\ifarxiv Figure~\ref{fig:staleness-litmus-tests} illustrates these forward and backward staleness conditions
on a collection of small example notebooks
(see also Section~\ref{sec:staleness}).\else The main paper illustrates these forward and backward staleness conditions
in Figure~\ref{fig:staleness-litmus-tests} on a collection of small example notebooks
(see also Section~\ref{sec:staleness} of the main paper).\fi{}
The remaining rules are straightforward.

\subsection{Well-Formedness}
\label{sec:well-formed}

The instrumented semantics maintains a \emph{well-formedness
  invariant} that characterizes when a notebook state permits
output-consistent execution.

\begin{definition}[Well-Formed State]
  \label{def:well-formed-supp}
  An instrumented state $S \cdot I = (C, O, \Sigma) \cdot (T, R, W)$ is
  \emph{well-formed} if for every $i$ with $T_i = \clean$, there
  exists $\Sigma'$ such that:
  \begin{enumerate}[nosep]
  \item $\InstEvalCell{C_i}{\Sigma}{O_i}{\Sigma'}{R_i}{W_i}$,
  \item $\Sigma$ and $\Sigma'$ agree except on
  $\bigcup W_{i+1..n}$, the set of locations written by cells below $i$.
  \item $R$ and $W$ are rerun consistent for $i$.
  \end{enumerate}
\end{definition}
That is, every clean cell $i$ can be re-executed from the current store to
produce its recorded output $O_i$ and read and write sets $R_i$ and $W_i$,
and only changes locations in $\Sigma'$ that are overwritten by writes below $i$.

\subsection{Correctness}
\label{sec:correctness}

We establish correctness by first showing that every notebook
operation preserves well-formedness. The three theorems below make
precise the preservation, output consistency, and progress guarantees
stated in Section~\ref{sec:guarantees};
their full proofs appear in
Sections~\ref{sec:preservation-proof}--\ref{sec:progress-proof}.

\begin{theorem}[Preservation]
  \label{thm:preservation}
  If $S \cdot I$ is well-formed and
  $S \cdot I \xRightarrow{op} S' \cdot I'$, then $S' \cdot I'$ is
  well-formed.
\end{theorem}

\begin{proof}[Proof sketch]
  The proof proceeds by cases on the operation. For
  \textsc{[Inst-Run]}, the key insight is that the rule antecedents
  ensure the executed cell satisfies rerun consistency, while the
  staleness predicates mark any cell that would violate rerun
  consistency. For \textsc{[Inst-Edit]}, the edited cell becomes
  stale, preserving rerun consistency vacuously. For structural
  operations, cells that would violate rerun consistency after index
  shifting are marked stale.
\end{proof}

Second, well-formed states with all cells clean are output consistent:

\begin{theorem}[Output Consistency]
  \label{thm:output-consistency}
  If $S \cdot I = (C, O, \Sigma) \cdot (T, R, W)$ is well-formed and
  $T_i = \clean$ for all $i$, then $(C, O, \Sigma)$ is output consistent, \emph{i.e.} there exists $\Sigma'$ such that
  $\StdEvalNotebook{C}{O}{\Sigma'}$.
\end{theorem}

\begin{proof}[Proof sketch]
  By induction on $i$. The inductive hypothesis states that executing
  $C_1, \ldots, C_i$ top-to-bottom from the empty store produces
  outputs $O_1, \ldots, O_i$ and a store that agrees with $\Sigma$ on
  locations in $W_{1..i}$ that do not conflict with writes below $i$
  ($W_{i+1..n}$). The base case is trivial. For the inductive step,
  rerun consistency of cell $i$ together with \WriteBeforeRead{} and
  \NoReadBeforeWrite{} ensure that the last execution of cell $i$ read
  the same values as it would in a top-to-bottom execution, thereby
  producing the same output $O_i$.
\end{proof}

Preservation and output consistency establish that well-formedness is
maintained and that all-clean states are output consistent, but not
that an all-clean state is \emph{reachable}. The natural strategy for
reaching one repeatedly runs the first stale cell. For
non-deterministic cells this strategy need not terminate: two cells with
empty read sets that alternately write $\{a\}$ and $\{b\}$ mark each
other stale forever. We therefore assume that cells are deterministic,
in a sequential sense. Evaluating code $c$ from store $\Sigma$ reads
locations one at a time; let $\rho(c, \Sigma)$ be the sequence of
locations it reads, up to completion or failure.

\begin{assumption}[Determinism]
  \label{asm:determinism}
  If $\InstEvalCell{c}{\Sigma}{o}{\Sigma'}{r}{w}$, then $r$ is the set of
  locations in $\rho(c, \Sigma)$. Moreover, for every store $\hat\Sigma$:
  \begin{enumerate}[nosep]
  \item if $\hat\Sigma$ agrees with $\Sigma$ on the first $k$
    locations of $\rho(c, \Sigma)$, then $\rho(c, \hat\Sigma)$ begins
    with the first $k+1$ locations of $\rho(c, \Sigma)$, or equals
    $\rho(c, \Sigma)$ if $k = |\rho(c, \Sigma)|$; and
  \item if $\hat\Sigma$ agrees with $\Sigma$ on $r$ and
    $\InstEvalCell{c}{\hat\Sigma}{\hat o}{\hat\Sigma'}{\hat r}{\hat w}$,
    then $\hat w = w$ and $\hat\Sigma'$ agrees with $\Sigma'$ on $w$.
  \end{enumerate}
\end{assumption}

\noindent
That is, the location a cell reads next, and the locations and values
it writes, depend only on the values it has read so far; its output may
still differ. Python executes a cell's statements in order, so every
deterministic cell satisfies Assumption~\ref{asm:determinism}.

The order of reads matters. Requiring only that a cell's read
\emph{set} depend on the values it reads does not ensure termination.
Consider three cells that write $a := 0$, $b := 0$, and $c := 0$,
followed by a cell $U$ that reads $\{a, b\}$ and writes $c := 1$ if
$a \neq 0 = b$; reads $\{b, c\}$ and writes $a := 1$ if
$b \neq 0 = c$; reads $\{c, a\}$ and writes $b := 1$ if
$c \neq 0 = a$; and reads all three otherwise. From the well-formed
state in which $a = 1$, $b = c = 0$, and only $U$ is stale, having last
written $a$, the strategy runs
$U$, then the writer of $a$, $U$, the writer of $c$, $U$, and the
writer of $b$, and returns to that state. Each run of $U$ stops writing
the location it wrote last, so \BackwardStale{} marks that location's
writer, whose rerun marks $U$ again. No cell that reads one location at
a time behaves like $U$: the location it reads first would belong to
every read set.

\begin{theorem}[Progress]
  \label{thm:progress}
  Under Assumption~\ref{asm:determinism}, if $S \cdot I$ is
  well-formed, then every sequence of \textsc{[Inst-Run]} steps from
  $S \cdot I$, each of which runs the first stale cell, is finite. When
  no such step exists, the current state $S' \cdot I'$ is well-formed
  and either:
  \begin{enumerate}[nosep]
  \item all cells are clean, so $S'$ is output consistent by
    Theorem~\ref{thm:output-consistency}; or
  \item the first stale cell $i$ is \emph{stuck}: there is no
    $S'' \cdot I''$ such that
    $S' \cdot I' \xRightarrow{\Run(i)} S'' \cdot I''$, because the rerun
    consistency checks of \textsc{[Inst-Run]} fail or the cell raises an
    exception.
  \end{enumerate}
\end{theorem}

\begin{proof}[Proof sketch]
  Every state reached is well-formed by
  Theorem~\ref{thm:preservation}, and a state whose first stale cell
  has no run is all-clean or stuck by definition, so the content of
  the theorem is termination. Let $E$ be the top-to-bottom execution of
  the current code, which Assumption~\ref{asm:determinism} makes unique
  up to outputs. Well-formedness implies that every cell above the first
  stale cell $f$ records $E$'s read and write sets, and that the store
  agrees with $E$ on the locations those cells write and no cell at or
  below $f$ writes. Hence a run of $f$ either behaves as in $E$ or first
  diverges from $E$ at a location $p$ that $f$ itself wrote. In the
  latter case, \NoReadAndWrite{} removes $p$ from $W_f$, and
  \BackwardStale{} marks the cell above $f$ whose rerun restores $p$.
  Runs of cells above $f$ behave as in $E$ and keep their write sets, so
  they move the first stale cell down, and each later run of $f$ agrees
  with $E$ on a strictly longer prefix of $E$'s read sequence for $f$.
  A lexicographic measure on these two quantities decreases with every
  run until the first stale cell moves below $f$.
\end{proof}

Together, these three theorems establish \flowbook's correctness:
starting from an initially well-formed state (all cells stale), every
sequence of operations maintains well-formedness
(Theorem~\ref{thm:preservation}); for deterministic cells, repeatedly
running the first stale cell terminates
(Theorem~\ref{thm:progress}); and when all cells are clean, the
notebook is output consistent
(Theorem~\ref{thm:output-consistency}).  The full proofs appear in
Sections~\ref{sec:preservation-proof}--\ref{sec:progress-proof} below.

\newcommand{\Eval}[4]{#1 \,,\, #2 \Downarrow #3 \,,\, #4}
\newcommand{\EvalFull}[6]{#1 \,,\, #2 \Downarrow #3 \,,\, #4 \,,\, #5 \,,\, #6}
\newcommand{\RunTo}[2]{#1 \xrightarrow{\Run(i)} #2}

\section{Preservation: Notebook Operations Preserve Well-Formedness}
\label{sec:preservation-proof}

In this section, we use $W_{i..j}$ to abbreviate
$W_i \cup \ldots \cup W_j$.

\noindent\textsc{Theorem~\ref{thm:preservation} (Preservation).}
\textit{If $S \cdot I$ is well-formed and
$S \cdot I \xRightarrow{op} S' \cdot I'$, then $S' \cdot I'$ is
well-formed.}

\begin{proof}
The proof proceeds by cases on the operation.

\begin{itemize}
\item \relname{Inst-Edit}.
$\Edit(i, c)$ sets $T'_i = \stale$ and leaves $R$ and $W$ unchanged.
Every cell that was $\clean$ in $I$ remains $\clean$ in $I'$ with the same witness.

\item \relname{Inst-Run}.
Suppose $S \cdot I = (C, O, \Sigma) \cdot (T, R, W)$ is well-formed and
\[
  S \cdot I \;\xRightarrow{\Run(i)}\; S' \cdot I' = (C, O', \Sigma') \cdot (T', R', W')
\]
where $R'$ and $W'$ agree with $R$ and $W$ except at position $i$.
We show $S' \cdot I'$ is well-formed.

Let $j$ be any index with $T'_j = \clean$.
We exhibit $\Sigma''$ such that
\[
  C_j;\, \Sigma' \,\Downarrow\, O'_j \cdot \Sigma'' \cdot R'_j \cdot W'_j
\]
with $\Sigma'$ and $\Sigma''$ agreeing except on $W_{j+1..n}$ and
where $R'$ and $W'$ are rerun consistent for $j$.

The proof splits into three subcases.

\begin{itemize}
\item \textsc{Subcase $j < i$.}
$\Run(i)$ sets $T'_i = \clean$ and may mark cells $\stale$ via
$\ForwardStale$ or $\BackwardStale$.  Since $T'_j = \clean$ and $j <
i$, cell $j$ was not marked $\stale$, so $T_j = \clean$.

Well-formedness of $S \cdot I$ yields $\Sigma'''$ such that
\[
  C_j;\, \Sigma \,\Downarrow\, O_j \cdot \Sigma''' \cdot R_j \cdot W_j
\]
where $\Sigma$ and $\Sigma'''$ agree except on $W_{j+1..n}$ and where
$R$ and $W$ are rerun consistent for $j$.

Since $j < i$, we have $R'_j = R_j$ and $W'_j = W_j$.
$\WriteBeforeRead(R', W', j)$ holds because $W'_{1..j-1} =
W_{1..j-1}$.  $\NoReadBeforeWrite(R', W', j)$ holds: any violation
would require $R_j \cap W'_i \neq \emptyset$, but then the
\[\NoWriteAfterRead(R', W', i)\] for $\Run(i)$ would fail.

The following diagram commutes:
\[
\begin{array}{ccc}
  \Sigma & \xrightarrow{\;C_j\;} & \Sigma''' \\[4pt]
  {\scriptstyle C_i}\downarrow \phantom{\scriptstyle C_i} & & \downarrow \\[4pt]
  \Sigma' & \xrightarrow{\;C_j\;} & \Sigma''
\end{array}
\]
The stores $\Sigma$ and $\Sigma'$ differ only on $W'_i$.
Since $R_j \cap W'_i = \emptyset$,
running $C_j$ from $\Sigma'$ yields the same behavior as from $\Sigma$.
Thus $\Sigma'$ and $\Sigma''$ agree except on $W_{j+1..n}$.
Hence $S' \cdot I'$ is well-formed.

\item \textsc{Subcase $j = i$.}
The rule \textsc{Inst-Run} sets $T'_i = \clean$.
The antecedent of the rule establishes that $R'$ and $W'$ are rerun consistent for $i$.

The evaluation $C_i;\, \Sigma \,\Downarrow\, O'_i \cdot \Sigma' \cdot
R'_i \cdot W'_i$ is given by the rule.  $\NoReadAndWrite(R', W', i)$
implies cell $i$ does not read its own writes.  Thus
$C_i;\, \Sigma' \,\Downarrow\, O'_i \cdot \Sigma' \cdot R'_i \cdot
W'_i$.  The store $\Sigma'$ agrees with itself except on
$W_{i+1..n}$.  Hence $S' \cdot I'$ is well-formed.

\item \textsc{Subcase $j > i$.}
Since $T'_j = \clean$ and $j > i$, cell $j$ was not marked $\stale$ by
$\ForwardStale(R', W', i, j)$.  Thus $W'_i \cap (R_j \cup W_j)
= \emptyset$ and $T_j = \clean$.

Well-formedness of $S \cdot I$ yields $\Sigma'''$ such that
\[
  C_j;\, \Sigma \,\Downarrow\, O_j \cdot \Sigma''' \cdot R_j \cdot W_j
\]
where $\Sigma$ and $\Sigma'''$ agree except on $W_{j+1..n}$ and where $R$ and $W$ are rerun consistent for $j$.

Since $j > i$, we have $R'_j = R_j$, $W'_j = W_j$, and $O'_j = O_j$, and
rerun consistency of $R'$ and $W'$ at $j$ holds.
The following diagram commutes:
\[
\begin{array}{ccc}
  \Sigma & \xrightarrow{\;C_j\;} & \Sigma''' \\[4pt]
  {\scriptstyle C_i}\downarrow \phantom{\scriptstyle C_i} & & \downarrow \\[4pt]
  \Sigma' & \xrightarrow{\;C_j\;} & \Sigma''
\end{array}
\]
The stores $\Sigma$ and $\Sigma'$ differ only on $W'_i$.
Since $W'_i \cap R_j = \emptyset$, running $C_j$ from $\Sigma'$ yields the same result.
Thus $\Sigma'$ and $\Sigma''$ agree except on $W'_{j+1..n}$, and well-formedness holds.
\end{itemize}

\item \relname{Inst-Insert}.
Suppose $S \cdot I \xRightarrow{\Insert(i, c)} S' \cdot I'$.
The rule sets $T'_i = \stale$ with $R'_i = \emptyset$ and $W'_i = \emptyset$.
All other cells shift indices: for $k \geq i$, cell $k$ in $I$ becomes cell $k+1$ in $I'$.

Let $j$ be any index with $T'_j = \clean$.
Then $j \neq i$ (since $T'_i = \stale$).
If $j < i$, cell $j$ is unchanged: $T'_j = T_j$, $R'_j = R_j$, $W'_j = W_j$.
If $j > i$, cell $j$ was cell $j-1$ in $I$: $T'_j = T_{j-1}$, $R'_j = R_{j-1}$, $W'_j = W_{j-1}$.

Since $W'_i = \emptyset$, inserting at $i$ does not affect $\WriteBeforeRead$ or $\NoReadBeforeWrite$.
The set $W'_{j+1..n}$ equals $W_{j+1..n}$ (for $j < i$) or $W_{j..n}$ (for $j > i$),
plus $W'_i = \emptyset$.
The witness store from $I$ serves as the witness in $I'$.

\item
\relname{Inst-Delete}.
Suppose $S \cdot I \xRightarrow{\Delete(i)} S' \cdot I'$.
The rule removes cell $i$ and marks cell $j$ stale if $R_j \cap W_i \neq \emptyset$ for all $j$.
For $k > i$, cell $k$ in $I$ becomes cell $k-1$ in $I'$.

Let $j$ be any index with $T'_j = \clean$.
Then $R_j \cap W_i = \emptyset$ (otherwise $j$ would be $\stale$).
If $j < i$, cell $j$ is unchanged in $R$ and $W$.
If $j \geq i$, cell $j$ was cell $j+1$ in $I$.

$\WriteBeforeRead(R', W', j)$ holds because $R_j \subseteq W_{1..j-1}$ and we removed $W_i$ where $R_j \cap W_i = \emptyset$.
Also, \[\NoReadBeforeWrite(R', W', j)\] holds.
The witness store from $I$ serves as the witness in $I'$: cell $j$ did not read $W_i$,
so removing those writes does not affect the evaluation.
Well-formedness holds.

\item
\relname{Inst-Move}.
$\Move(s, d)$ is defined as $\Delete(s)$ followed by $\Insert(d', C_s)$,
where $d' = d-1$ if $s < d$ and $d' = d$ if $s > d$.
Preservation follows by composing the two cases above:
if $S \cdot I$ is well-formed, then $\Delete(s)$ yields a well-formed $S'' \cdot I''$,
and $\Insert(d', C_s)$ yields a well-formed $S' \cdot I'$.
\end{itemize}

\medskip\noindent
All cases establish well-formedness of $S' \cdot I'$.
\end{proof}

\section{Output Consistency: Well-Formed All-Clean States Are Output Consistent}
\label{sec:output-consistency-proof}

\noindent\textsc{Theorem~\ref{thm:output-consistency} (Output Consistency).}
\textit{If $S \cdot I = (C, O, \Sigma) \cdot (T, R, W)$ is well-formed and
$T_i = \clean$ for all $i$, then there exists $\Sigma'$ such that
$\StdEvalNotebook{C}{O}{\Sigma'}$.}

\begin{proof}
The proof is by induction on $i$. Define
\[
  P(i) \;\defeq\; C_{1..i} \,\downarrow\, O_{1..i} \cdot \Sigma_i
  \quad\text{for some $\Sigma_i$ where $\Sigma$ and $\Sigma_i$ agree on } W_{1..i} \setminus W_{i+1..n}.
\]

\begin{itemize}
\item \textsc{Base case.}
$P(0)$ holds vacuously: $\varepsilon \,\downarrow\, \varepsilon \cdot \emptyset$.

\item \textsc{Inductive step.}
Suppose $P(i-1)$ holds.
Then $C_{1..i-1} \,\downarrow\, O_{1..i-1} \cdot \Sigma_{i-1}$
where $\Sigma$ and $\Sigma_{i-1}$ agree on $W_{1..i-1} \setminus W_{i..n}$.
We show $C_i;\, \Sigma_{i-1} \,\downarrow\, O_i \cdot \Sigma_i$.

Since $T_i = \clean$, well-formedness yields $\Sigma'$ such that
\[
  C_i;\, \Sigma \,\Downarrow\, O_i \cdot \Sigma' \cdot R_i \cdot W_i,
  \qquad
  \text{where $\Sigma$ and $\Sigma'$ agree except on } W_{i+1..n}.
\]
We have that:
\begin{align*}
  \NoReadAndWrite(R, W, i) & \implies R_i \cap W_i = \emptyset \\
  \WriteBeforeRead(R, W, i) &\implies R_i \subseteq W_{1..i-1} \\
  \NoReadBeforeWrite(R, W, i) &\implies R_i \cap W_{i..n} = \emptyset
\end{align*}
and:
\[
  R_i \subseteq W_{1..i-1} \setminus W_{i..n}
\]
By the inductive hypothesis, $\Sigma$ and $\Sigma_{i-1}$ agree on~$R_i$.
It follows that
\[
  C_i;\, \Sigma_{i-1} \,\downarrow\, O_i \cdot \Sigma_i
\]
Thus $\Sigma_i$ and $\Sigma$ agree on $W_{1..i} \setminus W_{i+1..n}$,
establishing $P(i)$.
\end{itemize}
\end{proof}

\section{Progress: Running Stale Cells Terminates}
\label{sec:progress-proof}

In this section, a \emph{run} is an \textsc{[Inst-Run]} step of the
first stale cell. Runs do not change the code $C$, so we can compare
every state with a single top-to-bottom execution of $C$.

\myparagraph{The reference execution}
Let $\Sigma^E_1 = \emptyset$. If $\Sigma^E_k$ is defined and
$\InstEvalCell{C_k}{\Sigma^E_k}{o}{\Sigma'}{r}{w}$ for some $o$, $\Sigma'$,
$r$, and $w$, then $\Sigma^E_{k+1} = \Sigma'$; otherwise
$\Sigma^E_{k+1}$ is undefined. By part~2 of
Assumption~\ref{asm:determinism}, all such evaluations write the same
values to the same locations and leave all other locations unchanged,
so $\Sigma^E_{k+1}$ is well defined; only their outputs may differ. Let
$\rho_k = \rho(C_k, \Sigma^E_k)$ be the read sequence of cell $k$ in
this execution, and let
\[
  X_k \;\defeq\; W_{1..k-1} \setminus W_{k..n}
\]
be the locations written above cell $k$ and not at or below it. Cell
$k$ is \emph{settled} if $\Sigma^E_k$ is defined,
$\InstEvalCell{C_k}{\Sigma^E_k}{o}{\Sigma'}{R_k}{W_k}$ for some $o$ and
$\Sigma'$, and $R_k \cap W_k = \emptyset$. That is, cell $k$ records the
read and write sets of the reference execution.

\begin{lemma}[Clean Prefix]
  \label{lem:clean-prefix}
  If $S \cdot I = (C, O, \Sigma) \cdot (T, R, W)$ is well-formed and
  $T_k = \clean$ for all $k < m$, then $\Sigma^E_m$ is defined, every
  cell $k < m$ is settled, and $\Sigma$ agrees with $\Sigma^E_m$ on $X_m$.
\end{lemma}

\begin{proof}
  By induction on $m$, following the proof of
  Theorem~\ref{thm:output-consistency}. For $m = 1$, $X_1 = \emptyset$.
  Suppose the lemma holds for $m$ and $T_m = \clean$. Rerun consistency
  of cell $m$ gives $R_m \subseteq X_m$, so $\Sigma$ and $\Sigma^E_m$
  agree on $R_m$, and the evaluation
  $\InstEvalCell{C_m}{\Sigma}{O_m}{\Sigma'}{R_m}{W_m}$ that
  well-formedness provides is also available from $\Sigma^E_m$, with
  the same read set, write set, and written values. Hence
  $\Sigma^E_{m+1}$ is defined, and cell $m$ is settled because
  \NoReadAndWrite{} gives $R_m \cap W_m = \emptyset$. Finally,
  $X_{m+1} = (W_m \setminus W_{m+1..n}) \cup (X_m \setminus W_m)$. On
  the first set, $\Sigma^E_{m+1}$ holds the values that cell $m$ writes,
  and so does $\Sigma$, which agrees with $\Sigma'$ outside
  $W_{m+1..n}$. On the second, $\Sigma^E_{m+1}$ agrees with
  $\Sigma^E_m$, and hence with $\Sigma$.
\end{proof}

\begin{lemma}[Run Analysis]
  \label{lem:run-analysis}
  Let $S \cdot I = (C, O, \Sigma) \cdot (T, R, W)$ be well-formed with
  first stale cell $f$, and let
  $S \cdot I \xRightarrow{\Run(f)} S' \cdot I'$, where cell $f$
  evaluates with read set $r$ and write set $w$. Then $\Sigma^E_f$ is
  defined, and either:
  \begin{enumerate}[nosep,label=(\alph*)]
  \item $\Sigma$ agrees with $\Sigma^E_f$ on $\rho_f$; then $r$ is the
    set of locations in $\rho_f$, and
    $\InstEvalCell{C_f}{\Sigma^E_f}{o'}{\Sigma''}{r}{w}$ for some $o'$ and
    $\Sigma''$; or
  \item $\rho_f = \pi \, p \, \pi'$, where $\Sigma$ agrees with
    $\Sigma^E_f$ on $\pi$ but not at $p$; then $r$ contains $p$ and the
    locations in $\pi$, and $p \in W_f \setminus w$.
  \end{enumerate}
\end{lemma}

\begin{proof}
  By Lemma~\ref{lem:clean-prefix}, $\Sigma^E_f$ is defined and $\Sigma$
  agrees with it on $X_f$. In case~(a), part~1 of
  Assumption~\ref{asm:determinism} gives $\rho(C_f, \Sigma) = \rho_f$,
  so $r$ is the set of locations in $\rho_f$, and because $\Sigma$ and
  $\Sigma^E_f$ agree on $r$, the evaluation from $\Sigma$ is also
  available from $\Sigma^E_f$. In case~(b), part~1 shows that
  $\rho(C_f, \Sigma)$ begins with $\pi \, p$. The checks of
  \textsc{[Inst-Run]} then give $p \in W_{1..f-1}$ (\WriteBeforeRead),
  $p \notin W_{f+1..n}$ (\NoReadBeforeWrite), and $p \notin w$
  (\NoReadAndWrite). Since $\Sigma$ and $\Sigma^E_f$ differ at $p$,
  $p \notin X_f$, so $p \in W_f$.
\end{proof}

\noindent\textsc{Theorem~\ref{thm:progress} (Progress).}
\textit{Under Assumption~\ref{asm:determinism}, if $S \cdot I$ is
well-formed, then every sequence of runs from $S \cdot I$ is finite.
When no run exists, the current state is well-formed and either all
cells are clean or the first stale cell is stuck.}

\begin{proof}
Every state reached is well-formed by
Theorem~\ref{thm:preservation}, and a state whose first stale cell has
no run is all-clean or stuck by definition. It remains to show that
every sequence of runs is finite.

We group runs into phases. A \emph{phase for cell $m$} begins in a
state whose first stale cell is $m$ and lasts while the first stale
cell is at or above $m$. By Lemma~\ref{lem:clean-prefix},
$\Sigma^E_m$ and hence $\rho_m$ are defined. For some $\kappa$ with
$0 \le \kappa \le |\rho_m| + 1$, every state of the phase satisfies:
\begin{enumerate}[nosep,label=(P\arabic*)]
\item every cell above $m$ is settled;
\item the first $\min(\kappa, |\rho_m|)$ locations of $\rho_m$ lie in
  $X_m$; and
\item if $\kappa = |\rho_m| + 1$, then
  $\InstEvalCell{C_m}{\Sigma^E_m}{o}{\Sigma'}{R_m}{W_m}$ for some $o$
  and $\Sigma'$.
\end{enumerate}
Initially $\kappa = 0$, and (P1) holds by
Lemma~\ref{lem:clean-prefix}. Consider a run of the first stale cell
$f \le m$, with read set $r$ and write set $w$.

\emph{Case $f < m$.} Cell $f$ is settled by (P1), so case~(b) of
Lemma~\ref{lem:run-analysis} cannot arise: it would place $p$ in $W_f$
and in $\rho_f$, whose locations form $R_f$ by
Assumption~\ref{asm:determinism}, but $R_f \cap W_f = \emptyset$. In case~(a), part~2 of
Assumption~\ref{asm:determinism}, applied to the evaluation of
Lemma~\ref{lem:run-analysis} and the one that makes $f$ settled, gives
$w = W_f$; also $r = R_f$. The run therefore leaves $R$ and $W$
unchanged, so (P1)--(P3) still hold. Since $W_f \setminus w = \emptyset$,
\BackwardStale{} marks nothing, and \ForwardStale{} marks only cells
below $f$, so the first stale cell moves below $f$.

\emph{Case $f = m$.} All cells above $m$ are clean, so by
Lemma~\ref{lem:clean-prefix} and (P2), $\Sigma$ agrees with
$\Sigma^E_m$ on the first $\min(\kappa, |\rho_m|)$ locations of
$\rho_m$. In case~(b) of Lemma~\ref{lem:run-analysis},
$\rho_m = \pi \, p \, \pi'$ with $|\pi| \ge \kappa$, and cell $m$ has
now read every location in $\pi \, p$, so the checks of
\textsc{[Inst-Run]} place them in $X_m$; (P2) and (P3) hold with
$|\pi| + 1 > \kappa$ in place of $\kappa$. In case~(a), if
$\kappa = |\rho_m| + 1$, then (P3) and part~2 of
Assumption~\ref{asm:determinism} give $w = W_m$, so no cell above $m$
is marked, every cell through $m$ is clean, and the phase ends.
Otherwise, (P2) and (P3) hold with $|\rho_m| + 1$ in place of $\kappa$:
the checks place the locations of $\rho_m$, which form $r$, in $X_m$, and
Lemma~\ref{lem:run-analysis} gives (P3). In both cases (P1) still
holds, since the run does not change the cells above $m$.

Every run in the phase therefore either ends the phase or decreases
$(|\rho_m| + 1 - \kappa,\; n - f)$ lexicographically. Hence the phase
contains at most $|\rho_m| + 2$ runs of cell $m$, each followed by at
most $m - 1$ runs of cells above $m$. A phase ends stuck or in a state
whose first stale cell, if any, lies below $m$, where the next phase
begins. There are thus at most $n$ phases, for strictly increasing
cells, and at most $\sum_m m\,(|\rho_m| + 2)$ runs in all.
\end{proof}

\section{Mechanization in Lean.}
\label{sec:lean}
\steve{We have mechanized the formal development in
  Lean~4, using only Lean's core library. The mechanization covers the
  notebook semantics and output consistency
  (Section~\ref{sec:semantics}), the instrumented semantics and
  well-formedness invariant (Section~\ref{sec:analysis-supp}), and
  complete proofs of the preservation, output-consistency, and
  progress theorems. It contains no \texttt{sorry}s and no custom
  axioms: every theorem depends only on Lean's three standard
  axioms. As in \ifarxiv these appendices\else this supplement\fi, cell evaluation is a black-box
  relation. The preservation and output-consistency proofs assume
  exactly two of its properties: evaluation does not change the store
  outside its write set, and evaluation depends only on the locations
  it reads. The progress proof additionally assumes
  Assumption~\ref{asm:determinism} and follows
  Section~\ref{sec:progress-proof}. To show that these assumptions are
  satisfiable, we instantiate them for a small concrete cell language
  and instantiate all three theorems at that model. We also check the
  example of Section~\ref{sec:correctness}: the cell $U$ satisfies the
  set-based form of determinism, yet running the first stale cell
  cycles forever.  The mechanization also includes
  an executable reference kernel for that language. Its rerun
  consistency checks and staleness propagation are proved to agree
  with the relational definitions, and each operation it accepts is
  proved to be a step of the instrumented semantics. The kernel
  replays every scenario of Figures~\ref{fig:litmus-tests}
  and~\ref{fig:staleness-litmus-tests} and reports the violations and
  staleness tags shown there.  The Lean proof is available at:
  \ifarxiv\url{https://github.com/stephenfreund/FlowBook}\else\url{https://anonymous.4open.science/r/FlowBook/}\fi.}